\documentclass[12pt]{article}

\usepackage{amssymb,amsmath,amsthm,amsfonts,amscd,bbm}
\usepackage{enumitem}
\usepackage{pdflscape}
\usepackage{caption}
\usepackage{bm}

\usepackage{ifpdf}
\ifpdf
\usepackage[pdftex]{graphicx}
\else
\usepackage[dvips]{graphicx}
\fi
\usepackage[all]{xy}
\usepackage{hyperref}
\usepackage{xcolor}
\usepackage{graphicx}
\usepackage[outdir=./]{epstopdf}
\usepackage{enumitem}
\usepackage{tensor}
\usepackage{tikz-cd}
\usepackage{multicol}
\usepackage{mathtools}
\usepackage{tocvsec2}
\usepackage{bbm}
\usepackage{longtable}
\usepackage{fullpage}

\newcommand{\doi}[1]{\href{https://dx.doi.org/#1}{{\small DOI:#1}}}

\newcommand{\googlebooks}[1]{(preview at \href{https://books.google.com/books?id=#1}{google books})}

\newcommand{\numdam}[1]{}

\usepackage{mathrsfs}
\DeclareMathAlphabet{\mathpzc}{OT1}{pzc}{m}{it}

\def\semicolon{;}
\def\applytolist#1{
    \expandafter\def\csname multi#1\endcsname##1{
        \def\multiack{##1}\ifx\multiack\semicolon
            \def\next{\relax}
        \else
            \csname #1\endcsname{##1}
            \def\next{\csname multi#1\endcsname}
        \fi
        \next}
    \csname multi#1\endcsname}

\def\calc#1{\expandafter\def\csname c#1\endcsname{{\mathcal #1}}}
\applytolist{calc}QWERTYUIOPLKJHGFDSAZXCVBNM;
\def\bbc#1{\expandafter\def\csname bb#1\endcsname{{\mathbb #1}}}
\applytolist{bbc}QWERTYUIOPLKJHGFDSAZXCVBNM;
\def\bfc#1{\expandafter\def\csname bf#1\endcsname{{\mathbf #1}}}
\applytolist{bfc}QWERTYUIOPLKJHGFDSAZXCVBNM;
\def\sfc#1{\expandafter\def\csname s#1\endcsname{{\sf #1}}}
\applytolist{sfc}QWERTYUIOPLKJHGFDSAZXCVBNM;
\def\fc#1{\expandafter\def\csname f#1\endcsname{{\mathfrak #1}}}
\applytolist{fc}QWERTYUIOPLKJHGFDSAZXCVBNM;

\usepackage{tikz}
\usepackage{tikz-cd}

\makeatletter
\def\fixtikzforbreqn#1#2{%
  \protected\edef#1{\noexpand\ifmmode\mathchar\the\mathcode`#2 \noexpand\else#2\noexpand\fi}%
}
\fixtikzforbreqn\tikz@nonactivesemicolon;
\fixtikzforbreqn\tikz@nonactivecolon:
\fixtikzforbreqn\tikz@nonactivebar|
\fixtikzforbreqn\tikz@nonactiveexlmark!
\makeatother

\usepackage{breqn}   

\usetikzlibrary{arrows,backgrounds,patterns.meta}
\usetikzlibrary{positioning,shadings,cd}
\usetikzlibrary{shapes}
\usetikzlibrary{backgrounds}
\usetikzlibrary{decorations,decorations.pathreplacing,decorations.markings,decorations.pathmorphing}
\usetikzlibrary{fit,calc,through}
\usetikzlibrary{external}
\usetikzlibrary{arrows}
\tikzset{vertex/.style = {shape=circle,draw,fill=black,inner sep=0pt,minimum size=5pt}}
\tikzset{edge/.style = {->,> = latex', bend right}}
\tikzset{
	super thick/.style={line width=3pt}
}
\tikzset{
    quadruple/.style args={[#1] in [#2] in [#3] in [#4]}{
        #1,preaction={preaction={preaction={draw,#4},draw,#3}, draw,#2}
    }
}
\tikzstyle{shaded}=[fill=red!10!blue!20!gray!30!white]
\tikzstyle{unshaded}=[fill=white]
\tikzstyle{empty box}=[circle, draw, thick, fill=white, opaque, inner sep=2mm]
\tikzstyle{annular}=[scale=.7, inner sep=1mm, baseline]
\tikzstyle{rectangular}=[scale=.75, inner sep=1mm, baseline=-.1cm]
\tikzstyle{mid>}=[decoration={markings, mark=at position 0.5 with {\arrow{>}}}, postaction={decorate}]
\tikzstyle{mid<}=[decoration={markings, mark=at position 0.5 with {\arrow{<}}}, postaction={decorate}]
\tikzstyle{over}=[double, draw=white, super thick, double=]
\tikzstyle{snake}=[decorate, decoration={snake, segment length=1mm, amplitude=.3mm}]
\tikzstyle{saw}=[decorate, decoration={saw, segment length=.7mm, amplitude=.25mm}]

\tikzstyle{coupon}=[draw, very thick, rectangle, rounded corners=5pt]
\tikzset{Rightarrow/.style={double equal sign distance,>={Implies},->},
triplecd/.style={-,preaction={draw,Rightarrow}},
quadruplecd/.style={preaction={draw,Rightarrow,
shorten >=0pt
},
shorten >=1pt,
-,double,double
distance=0.2pt}}
\tikzset{
    tripleline/.style args={[#1] in [#2] in [#3]}{
        #1,preaction={preaction={draw,#3},draw,#2}
    }
}
\tikzstyle{triple}=[tripleline={[line width=.15mm,black] in
      [line width=.7mm,white] in
      [line width=1mm,black]}] 
\tikzset{
    quadrupleline/.style args={[#1] in [#2] in [#3] in [#4]}{
        #1,preaction={preaction={preaction={draw,#4},draw,#3}, draw,#2}
    }
}
\tikzstyle{quadruple}=[quadrupleline={[line width=.3mm,white] in
      [line width=.6mm,black] in
      [line width=1.2mm,white] in
      [line width=1.5mm,black]}]

\theoremstyle{plain}
\newtheorem{thm}{Theorem}[section]
\newtheorem*{thm*}{Theorem}
\newtheorem{thmalpha}{Theorem}

\newtheorem{cor}[thm]{Corollary}
\newtheorem{coralpha}[thmalpha]{Corollary}
\newtheorem*{cor*}{Corollary}

\newtheorem*{conj*}{Conjecture}
\newtheorem{lem}[thm]{Lemma}
\newtheorem*{lem*}{Lemma}

\newtheorem{prop}[thm]{Proposition}

\newtheorem{problem}[thm]{Problem}
\newtheorem*{quest*}{Question}
\newtheorem*{claim*}{Claim}

\theoremstyle{definition}
\newtheorem{defn}[thm]{Definition}

\newtheorem{ex}[thm]{Example}
\newtheorem{sub-ex}[thm]{Sub-Example}
\newtheorem{counter-ex}[thm]{Counter-Example}
\newtheorem{rem}[thm]{Remark}
\newtheorem*{rem*}{Remark}

\usepackage{xcolor}
\definecolor{dark-red}{rgb}{0.7,0.25,0.25}
\definecolor{dark-blue}{rgb}{0.15,0.15,0.55}
\definecolor{medium-blue}{rgb}{0,0,.8}
\definecolor{DarkGreen}{RGB}{0,150,0}
\definecolor{rho}{named}{red}
\hypersetup{
   colorlinks, linkcolor={purple},
   citecolor={medium-blue}, urlcolor={medium-blue}
}

\def\altdb{\vadjust{\vbox to 0pt{\vss\hbox{\kern \hsize
\quad{\dbend}}\kern\baselineskip\kern-10pt}}}

\newcommand{\noshow}[1]{}

\begin{document} 

\title{DHR bimodules of quasi-local algebras and symmetric quantum cellular automata }
\author{Corey Jones}

\maketitle

\begin{abstract}
For a net of C*-algebras on a discrete metric space, we introduce a bimodule version of the DHR tensor category and show that it is an invariant of quasi-local algebras under isomorphisms with bounded spread. For abstract spin systems on a lattice $L\subseteq \mathbbm{R}^{n}$ satisfying a weak version of Haag duality, we construct a braiding on these categories. Applying the general theory to quasi-local algebras $A$ of operators on a lattice invariant under a (categorical) symmetry, we obtain a homomorphism from the group of symmetric QCA to $\textbf{Aut}_{br}(\textbf{DHR}(A))$, containing symmetric finite depth circuits in the kernel. For a spin chain with fusion categorical symmetry $\mathcal{D}$, we show that the DHR category of the quasi-local algebra of symmetric operators is equivalent to the Drinfeld center $\mathcal{Z}(\mathcal{D})$ . 
We use this to show that for the double spin flip action $\mathbbm{Z}/2\mathbbm{Z}\times \mathbbm{Z}/2\mathbbm{Z}\curvearrowright \mathbbm{C}^{2}\otimes \mathbbm{C}^{2}$, the group of symmetric QCA modulo symmetric finite depth circuits in 1D contains a copy of $S_{3}$, hence it is non-abelian, in contrast to the case with no symmetry.
\end{abstract}

\tableofcontents

\section{Introduction}

In the algebraic approach to quantum spin systems on a lattice, a fundamental role is played by the quasi-local C*-algebra generated by local operators \cite{MR1441540}. In ordinary spin systems, this is an infinite tensor product of matrix algebras. Upon restricting to operators invariant under a global (categorical) symmetry or when considering the operators acting on the boundary of a topologically ordered spin system, the resulting quasi-local algebras can be more complicated approximately finite dimensional (AF) algebras. The work of Bratteli \cite{MR312282} and Elliott \cite{MR397420} gives a classification of AF algebras up to isomorphism. However, arbitrary isomorphisms between quasi-local algebras are not always physically relevant since, in general, they do not map local Hamiltonians to local Hamiltonians\footnote{Locality means many different things in different contexts. Here, by local Hamiltonian we mean the terms in the Hamiltonian have supports with uniformly bounded diameters \cite[Chapter 4]{MR3929747}.} 

A physically natural condition to impose on isomorphisms between quasi-local algebras defined on the same metric space is \textit{bounded spread}. For two nets of algebras defined on a discrete metric space $L$, an isomorphism $\alpha$ between their quasi-local algebras has bounded spread if there exists an $R\ge 0$ such that operators localized in a finite region $F\subseteq L$ are mapped to operators localized in the $R$ neighborhood of $F$ by $\alpha$ and $\alpha^{-1}$. Unlike generic isomorphisms between quasi-local algebras, isomorphisms with bounded spread map local Hamiltonians to local Hamiltonians, making them more natural from a physical perspective. This raises the general problem of classifying quasi-local algebras up to bounded spread isomorphism.

Bounded spread isomorphisms are also interesting as objects in their own right. Automorphisms of the quasi-local algebra of a spin system (without symmetry) with bounded spread are called \textit{quantum cellular automata} (QCA) \cite{ https://doi.org/10.48550/arxiv.quant-ph/0405174}, and have been extensively studied in the physics literature (we refer the reader to the review article \cite{Farrelly2020reviewofquantum} and references therein). These can be viewed as a natural class of symmetries of the moduli space of local Hamiltonians, but also are natural models for discrete-time unitary dynamics. Finite depth quantum circuits (FDQC) are a normal subgroup of QCA which are implemented by local unitaries, and are used as to operationally define equivalence for topologically ordered states \cite{PhysRevB.82.155138}. There is significant interest in understanding the quotient group QCA/FDQC, which can be interpreted as the collection of topological phases of QCA \footnote{see Section \ref{physical interpretation} for further discussion and references.} \cite{MR2890305, MR4103966, MR4381173, MR4544190, https://doi.org/10.48550/arxiv.2211.02086,https://doi.org/10.48550/arxiv.2205.09141, PRXQuantum.3.030326}. While there has been recent progress on studying symmetry protected QCA \cite{IgnacioCirac_2017, PhysRevLett.124.100402}, relatively little is known about the structure of topological phases of QCA defined only on symmetric operators. 

We can approach both the problem of finding bounded spread isomorphism invariants of quasi-local algebras and of finding invariants of QCA/FDQC simultaneously, by looking for \textit{functorial} invariants of quasi-local algebras. To be more precise, consider the groupoid $\textbf{Net}_{L}$ whose objects are general nets of C*-algebras on a discrete metric space $L$ (Definition \ref{Discretenetdef}), and whose morphisms are isomorphisms of quasi-local algebras with bounded spread. Then any functor from $\text{Net}_{L}$ to another groupoid which contains finite depth circuits in the kernel will yield algebraic invariants of both general quasi-local algebras and of topological phases  of QCA.

An important component of an algebraic quantum field theory is its DHR category of superselection sectors \cite{MR258394, MR297259, MR334742}. Motivated by finding functorial invariants for discrete nets of C*-algebras, we develop a version of DHR theory suitable for our setting. For a net of C*-algebras $A$ over a discrete metric space $L$, we introduce the C*-tensor category $\textbf{DHR}(A)$, which consists of localizable bimodules of the quasi-local algebra (Definition \ref{localizable bimodule}). This is a direct generalization of localized, transportable endomorphisms from the usual DHR formalism \cite{MR1405610, Halvorson2006-HALAQF}. Our formalism extends the ideas of \cite{MR1463825}, who consider the special case of 1D spin chains with Hopf algebra symmetry and utilize the formalism of unital amplimorphisms rather than bimodules.

To state the first main result of the paper, let $\textbf{C*-Tens}$ denote the groupoid of C*-tensor categories and unitary tensor equivalences (up to unitary monoidal natural isomorphism). Then we have the following theorem:

\begin{thmalpha}\label{ThmA}
Let $L$ be a countable metric space with bounded geometry. There is a canonical functor $\textbf{DHR}: \textbf{Net}_{L}\rightarrow \textbf{C*-Tens}$, containing finite depth quantum circuit in the kernel. In particular

\begin{enumerate}
\item 
The monoidal equivalence class of $\textbf{DHR}(A)$ is an invariant of the quasi-local algebra up to bounded spread isomorphism.
\item
We have a homomorphism 
$$\textbf{DHR}: \textbf{QCA}(A)/\textbf{FDQC}(A)\rightarrow \textbf{Aut}_{\otimes}(\textbf{DHR}(A)).$$

\end{enumerate}
\end{thmalpha}

The first item above allows us to distinguish quasi-local algebras that are isomorphic as C*-algebras but not by bounded spread isomorphisms, while the second gives us a topological invariant of QCA. In particular, we can conclude that a QCA is \textit{not} a quantum circuit if it has a non-trivial image in $\textbf{Aut}_{\otimes}(\textbf{DHR}(A))$. We will exploit both of these consequences in the case of 1D symmetric spin systems (see Examples \ref{genspinflip} and \ref{non-abelian}). 

First, we address the issue of braidings. In the usual DHR theory the resulting categories are braided, which plays a significant role in many applications. In our context, this additional structure provides a finer invariant for quasi-local algebras and restricts the image of the DHR homomorphisms from QCA. Under some additional assumptions on the lattice (namely, that it is a discrete subspace of $\mathbbm{R}^{n}$) and the net itself (weak algebraic Haag duality, Definition \ref{algebraicHaag}), our DHR categories admit canonical braidings, and bounded spread isomorphisms induce braided equivalences on DHR categories.

\begin{thmalpha}\label{ThmB} Suppose $L\subseteq \mathbbm{R}^{n}$ is a lattice. If a net $A$ over $L$ satisfies weak algebraic Haag duality, there exists a canonical braiding on $\textbf{DHR}(A)$. Furthermore, if $A$ and $B$ are two such nets, then for any isomorphism $\alpha:A\rightarrow B$ with bounded spread, $\textbf{DHR}(\alpha)$ is a braided equivalence. As a consequence, we obtain \begin{enumerate}
\item 
The \textit{braided} monoidal equivalence class of $\textbf{DHR}(A)$ is an invariant of the quasi-local algebra up to bounded spread isomorphism.
\item
We have a homomorphism $\textbf{DHR}: \textbf{QCA}(A)/\textbf{FDQC}(A)\rightarrow \textbf{Aut}_{br}(\textbf{DHR}(A))$.

\end{enumerate}
\end{thmalpha}

We proceed to apply the general theory to the case of 1D spin systems with fusion categorical symmetry. Categorical symmetry can be formalized either in terms of matrix product operators (MPOs) or weak Hopf algebra actions. In either case, we can realize the quasi-local algebra of symmetric operators as a net over $\mathbbm{Z}$, where the local algebras are endomorphisms of tensor powers of an object $X$ internal to a unitary fusion category $\cD$. Recall that $\mathcal{Z}(\cD)$ denotes the \textit{Drinfeld center} of $\cD$.

 \begin{thmalpha}\label{ThmC}
Let $\cD$ be a unitary fusion category and suppose $X\in \cD$ is strongly tensor generating. Then the net $A$ over $\mathbbm{Z}\subseteq \mathbbm{R}$ of tensor powers of $X$ satisfies (weak) algebraic Haag duality, and $\textbf{DHR}(A)\cong \mathcal{Z}(\cD)$ as braided C*-tensor categories. In particular, there exists a canonical homomorphism  $$\textbf{DHR}: \textbf{QCA}(A)/\textbf{FQDC}(A)\rightarrow \textbf{Aut}_{br}(\mathcal{Z}(\cD))\cong \textbf{BrPic}(\cD).$$

Furthermore, if $X$ is a characteristic object\footnote{we call an object characteristic if it is fixed up to isomorphism by any monoidal autoequivalence}, then the image of $\textbf{DHR}$ contains the subgroup $\textbf{Out}(\cD)\subseteq \textbf{Aut}_{br}(\mathcal{Z}(\cD))$.
 \end{thmalpha}

The equivalence of the DHR category with the Drinfeld center generalizes the main result of \cite{MR1463825} from the context of Hopf algebra symmetries to general fusion categorical symmetry. This family of categorical nets was recently studied from a physical perspective in \cite{https://doi.org/10.48550/arxiv.2112.09091, https://doi.org/10.48550/arxiv.2211.03777}. In these works, bounded spread isomorphisms between nets are constructed from categorical data which implement duality transformations on symmetric Hamiltonians using matrix product operators. A key role is played by their notion of topological sector, which we expect to be closely related to our DHR bimodules.   

Our analysis of $\textbf{DHR}(A)$ makes heavy use of the techniques of subfactor theory \cite{MR1642584,MR1473221,MR1339767,math.QA/9909027} recently translated to the C*-context \cite{MR4419534, https://doi.org/10.48550/arxiv.2207.11854} (see Section \ref{QuantumSymms}). We refer the reader to \cite{MR1463825, 2205.15243, MR4272039} for a related analysis of 1D spin systems from a subfactor point of view. 

One of the most remarkable results in the theory of QCA is that the group QCA/FDQC of an ordinary spin system is abelian, even without adding ancilla \cite{MR4381173}. As a corollary of our results, we will see that in the symmetric case this is not true. First consider an ordinary spin system, where the local Hilbert space is $\mathbbm{C}^{2}$ with the $\mathbbm{Z}/2\mathbbm{Z}$ spin flip symmetry. We partition the system into adjacent pairs and coarse grain so that the local Hilbert space is $K:=\mathbbm{C}^{2}\otimes \mathbbm{C}^{2}$, and the group is $\mathbbm{Z}/2\mathbbm{Z}\times \mathbbm{Z}/2\mathbbm{Z} $ acting on $K$ by a ``double spin flip".
 
 \begin{coralpha}\label{CorD}
    For the double spin flip $\mathbbm{Z}/2\mathbbm{Z}\times \mathbbm{Z}/2\mathbbm{Z}\curvearrowright \mathbbm{C}^{2}\otimes \mathbbm{C}^{2}$ on-site symmetry, the group of symmetric QCA modulo symmetric finite depth circuits contains $S_{3}$ and in particular is non-abelian. 
 \end{coralpha}

It is clear that $\textbf{DHR}$ is not a complete invariant for QCA up to finite depth circuits even in 1D. Indeed, for the case of the trivial categorical symmetry, this is an ordinary spin system and our invariant is trivial. However, the group QCA/FDQC is a highly non-trivial subgroup of $\mathbbm{Q}^{\times}$, with isomorphism given by the GNVW index \cite{MR2890305}. However, we believe the action on the DHR category will be the crucial component beyond index theory for any general classification scheme for symmetric QCA.

Finally, while we have motivated our DHR theory with applications to understanding isomorphisms between quasi-local algebras with bounded spread, we anticipate many further applications. For example, for any state $\phi$ on a quasi-local algebra $A$, the superselection category of $\phi$ is a module category over $\textbf{DHR}(A)$, opening the door to an intrinsically categorical (rather than analytic) treatment of superselection theory of states for symmetric nets. In another direction, we believe that discrete nets of C*-algebras over a (sufficiently regular) fixed lattice in $\mathbbm{R}^{n}$ should assemble into a symmetric monoidal $n+2$ category, with the $n=1$ case being a discrete version of the symmetric monoidal 3-category of coordinate-free CFTs \cite{MR3439097, MR3927541, MR3773743}. The DHR category of a net $A$ we consider here should then arise as $\Omega^{n+1}(A)$ in the $n+2$ category.

\bigskip

\textbf{Acknowledgements}. I would like to thank the entire ``QCA group" from the AIM workshop ``Higher categories and topological order" for many stimulating discussions which ultimately sparked the ideas for this paper. Thanks in particular to Jeongwan Haah for teaching me about the general theory and motivation for QCA, and Dom Williamson for suggesting both the problem of studying symmetric QCA and a possible relationship with symmetries of the Drinfeld center. Also thanks to Dave Aasen, Jacob Bridgeman, Peter Huston, Laurens Lootens, Pieter Naaijkens, Dave Penneys, David Reutter and Daniel Wallick for many enlightening discussions and helpful comments on early drafts of this paper. Finally, I want to thank Vaughan Jones for always encouraging me to look for the physics in mathematics. This work was supported by NSF Grant DMS-2100531.

\section{Discrete nets of C*-algebras}

In this section we introduce our general mathematical framework, which is a straightforward ``AQFT-style" extension of the usual axioms for abstract spin systems as found, for example, in \cite{MR1441540}. These mathematical objects are meant to axiomatize the algebras of local operators of any kind of discrete quantum field theory, which simultaneously encodes both local observables and local unitaries. The version of discrete metric space which we found most appropriate for our framework is the following:

\begin{defn} We say a countably infinite metric space $L$ has bounded geometry if for any $R\ge 0$, there exists an $S$ with $|B_{R}(x)|\le S$ for all $x\in L$.

\end{defn}

In the above definition, we use the notation $B_{R}(x)$ to denote the (closed) ball of radius $R$ about the point $x$. Also note that in the above definition, we are assuming our space is countably infinite. Examples include: Cayley graphs of infinite, finitely generated groups (or more generally path metrics on graphs with bounded degree), and discrete subsets of Riemannian manifolds with bounded sectional curvature. Bounded geometry spaces play an important role in the study of large-scale geometry (see \cite{MR2986138}).

We denote the \textit{poset of finite subsets ordered by inclusion} in $L$ by $\mathcal{F}(L)$, and the \textit{poset of balls ordered by inclusion} by $\mathcal{B}(L)$. These will be the fundamental ``small regions" in our discrete QFT.

\begin{defn}\label{Discretenetdef}
A \textit{discrete net} of C*-algebras consists of an infinite bounded geometry metric space $L$, a unital C*-algebra $A$ (called the quasi-local algebra), and a poset homomorphism from $\mathcal{F}(L)$ to the poset of unital C*-subalgebras of $A$ ordered by inclusion, denoted $F\mapsto A_{F}$, subject to the following conditions:
\begin{enumerate}
\item
If $\displaystyle F\bigcap G=\varnothing$, then $[A_{F},A_{G}]=0$.
\item
$\displaystyle \bigcup_{F\in \mathcal{F}(L)} A_{F}$ is dense in $A$.
\end{enumerate}
\end{defn}

To simplify the notation, we will often simply denote a discrete net in terms of its quasi-local algebra $A$, with the additional structure of the poset homomorphism from finite subsets of $L$ to unital subalgebras of $A$ implicit additional structure.

We note that we can naturally extend our poset homomorphism from the poset of finite subsets to $\mathcal{P}(L)$, the collection of \textit{all} subsets of $L$, as follows:

For any $M\subseteq L$, define 

$$A_{M}:=\text{C*}\langle \{x\in A_{F}\ : F\in \mathcal{F}(L)\ \text{and}\ F\subseteq M\} \rangle$$

\noindent In other words, $A_{M}$ is the C* subalgebra of $A$ generated by the algebras $A_{F}$, where $F$ ranges over finite subsets of $M$. The two requirements in the definition for a discrete net now hold when replacing $\mathcal{F}(L)$ with $\mathcal{P}(L)$. 

We can also use other data to generate a net. For example, we may have a poset homomorphism from the poset $\mathcal{B}(L)$ of balls in $L$ to subalgebras of $A$, and we can extend this to be defined on $\mathcal{P}(L)$ (and hence on $\mathcal{F}(L)$) in the same way. In practice, this is usually how we will do things, but there is nothing really special about balls, and other types of standard regions (e.g. rectangles) work equally as well.

\begin{ex}{\textbf{Spin systems}}. The fundamental family of examples are the nets of spin observables. Let $L$ be an arbitrary metric space with bounded geometry. Fix a positive integer $d$ and define $A^{d}$ to be the UHF algebra $\displaystyle M_{d^{\infty}}\cong \otimes_{x\in L} M_{d}(\mathbbm{C})$, where here $M_{d}(\mathbbm{C})$ denotes the algebra of $d\times d$ matrices. For each finite subset in $\mathcal{F}(L)$, we set $A^{d}_{F}:=\otimes_{x\in F} M_{d}(\mathbbm{C})\subseteq A^{d}$. This clearly satisfies the axioms of a discrete net. For an extensive exposition on this class of examples, see \cite{MR1441540}.
\end{ex}

\begin{ex}{\textbf{Symmetric spin systems}}. Suppose we start with a spin system $A$ over $L$ equipped with a global, onsite symmetry $G$. More specifically, suppose we have a homomorphism $G\rightarrow \text{Aut}(M_{d}(\mathbbm{C}))$, where $d$ is the dimension of the on-site Hilbert space. Then by taking the infinite tensor product, this defines a global symmetry on $A^{d}$ which preserves the local subalgebras. We set $A^{G}$ to be the algebra of operators invariant under the $G$ action, and for any ball $F\in \mathcal{B}(L)$, we set $A^{G}_{F}:=(A^{d}_{F})^{G}$. This assembles into a discrete net over $L$ as discussed above, and serves as the motivating example of a discrete net that is of physical interest but not an ordinary spin system. By forcing invariance under $G$, we are implementing \textit{local} superselection sectors. One of the goals of this paper is to give a model independent formulation of these superselection sectors via a DHR category.

There are many generalizations of group symmetry that are considered in the context of spin systems. For example, in 1D we can have fusion categorical (or weak Hopf algebra) symmetries, and taking invariant local operators gives us a new net. We will study such examples in depth in Section \ref{categorical symmetry}. In the world of AQFT, taking the net of operators invariant under a global symmetry is sometimes called \textit{gauging the global symmetry}, or applying the orbifold construction. We encourage the reader to think of an abstract discrete net as a gauging of a spin system by some kind of (possibly generalized) global symmetry, so that the elements in $A_{F}$ are the operators that are invariant under a global symmetry. 

\end{ex}

\begin{ex}\label{BoundaryEx}{\textbf{Boundaries of commuting projector systems }.} Consider a commuting projector Hamiltonian an on the regular lattice $\mathbbm{Z}^{n}$. Consider the half-lattice $\mathbbm{Z}^{n}\le 0$, which has a boundary lattice equivalent to $\mathbbm{Z}^{n-1}$. Define a net of algebras on $\mathbbm{Z}^{n-1}$ consisting of operators localized near the boundary, and cut down by the projection $P$ onto the bulk ground state, similarly to \cite{Ha16}. Modulo some technical details, this assembles into a net of ``boundary algebras" which can have non-trivial local superselection sectors. Applying the DHR construction from Section \ref{DiscreteDHRTheory} to the boundary quasi-local algebra yields a braided tensor category, which should correspond to the topological order of the bulk theory. This is a concrete manifestation of a ``bulk-boundary correspondence" in the setting of topological codes. We will clarify this story in future work.

\end{ex}

\bigskip

For any subset $F\in \mathcal{P}(L)$ and $R\ge 0$, we define its $R$-neighborhood 

$$N_{R}(F):=\{x\in L\ :\ d(x,F)\le R\}.$$

A property that may be satisfied by discrete nets that will sometimes be useful is the following.

\begin{defn}\label{boundedlygenerated} A discrete net is \textit{boundedly generated} if there exists an $T\ge 0$ such every $A_{F}$ is generated by its subalgebras $\{A_{G}\ :\ G\subseteq F\ \text{and}\ \text{diam}(G)< T\}$.
\end{defn}

This condition guarantees that the algebra is generated ``uniformly locally".  This is a weak version of an additivity-type axiom in AQFT. We do not need to assume it for any technical results, but it is a nice property that the nets in our examples will always satisfy.

We now move on to define a technical condition that will be fundamental for building a braiding on discrete DHR categories. Recall that if $B\subseteq A$ is a subset of the algebra $A$, the centralizer of $B$ in $A$ is defined as
$$Z_{A}(B):=\{x\in A\ :\ [x,y]=0\ \text{for all}\ y\in B\}.$$

\begin{defn}\label{algebraicHaag}(c.f. \cite[Definition 2.3]{MR1463825}) A discrete net $A$ satisfies 

\begin{enumerate}
    \item 
\textit{weak algebraic Haag duality} if there exists $R, D\ge 0$ such that for any $F\in \mathcal{B}(L)$ of radius $U\ge R$ about the point $x\in L$, $Z_{A}(A_{F^{c}})\subseteq A_{G}$, where $G\in \mathcal{B}(L)$ is the ball about $x$ of radius $U+D$. Specific choices of $R$ and $D$ are called \textit{duality constants}.
\item
\textit{algebraic Haag duality} if it satisfies weak algebraic Haag duality with $D=0$. In this case $Z_{A}(A_{F^{c}})=A_{F}$.
\end{enumerate}
\end{defn}

\begin{rem} Algebraic Haag duality is a version of the usual Haag duality from AQFT \cite{MR1405610}, with the major difference that we are only asking for the \textit{relative commutant} of the $A_{F^{c}}$ in $A$ to be $A_{F}$, rather than the commutant in a larger $B(H)$ for some global Hilbert space $H$. Weak algebraic Haag duality is inspired by the weak Haag duality of Ogata, used to derive braided categories in the context of topologically ordered spin systems \cite{MR4362722}. All of our examples of interest satisfy algebraic Haag duality, but the weaker version has the added theoretical advantage of being invariant under isomorphisms with bounded spread, which we show below. We thank Pieter Naaijkens, David Penneys, and Daniel Wallick for discussions on the related topic of topologically ordered spin systems, where a similar version of weak Haag duality emerged naturally.
\end{rem}

Conceptually, weak algebraic Haag duality gives us a powerful tool to verify an operator is localized in a finite region by checking that it commutes with all operators localized in the complement. 

\subsection{Bounded spread isomorphisms and QCA}

\begin{defn}\label{boundespread} For two discrete nets $A$ and $B$ over the metric space $L$, a $*$-isomorphism $\alpha:A\rightarrow B$ of quasi-local algebras has \textit{bounded spread} if there exists an $R\ge 0$ such that for any $F\in \mathcal{F}(L)$, $\alpha(A_{F})\subseteq B_{N_{R}(F)}$ and $\alpha^{-1}(B_{F})\subseteq A_{N_{R}(F)} $
\end{defn}

\begin{defn} For a fixed infinite metric space $L$ with bounded geometry, $\textbf{Net}_{L}$ is the groupoid  whose 

\begin{enumerate}
    \item 
Objects are discrete nets over $L$.
\item
Morphisms $\textbf{Net}_{L}(A,B)$ consist of $*$-isomorphisms $\alpha: A\rightarrow B$ such that $\alpha$ has bounded spread.
\end{enumerate}
\end{defn}

In many examples, $\alpha(A_{F})\subseteq B_{N_{R}(F)}$ for all $F$ automatically implies $\alpha^{-1}(B_{F})\subseteq A_{N_{R}(F)} $ (for example, in ordinary spin systems \cite{ARRIGHI2011372}).

\begin{prop}
The property of weak algebraic Haag duality is invariant under bounded spread isomorphism.
\end{prop}

\begin{proof}
    Suppose $A$ satisfies weak algebraic Haag duality, with constants $R$ and $D$, and suppose $\alpha: A\rightarrow B$ is a *-isomorphism with spread at most $T$. We claim $B$ satisfies weak algebraic Haag duality with constants $R, D+2T$. Let $F$ be a ball of radius $U\ge R$ about some point $x$. Then set $F^{\prime}$ to be the corresponding ball of radius $U+T$ and $F^{\prime \prime}$ the ball of radius $U+T+D$. Then $A_{(F^{\prime})^{c}}\subseteq \alpha^{-1}(B_{F^{c}})$, so

    \begin{align*}
    \alpha^{-1}(Z_{B}(B_{F^{c}}))&=Z_{A}(\alpha^{-1}(B_{F^{c}}))\\
    &\subseteq Z_{A}(A_{(F^{\prime})^{c}})\\
    &\subseteq A_{F^{\prime \prime}}
    \end{align*}

    Therefore $$Z_{B}(B_{F^{c}})\subseteq \alpha(A_{F^{\prime \prime}})\subseteq B_{G},$$

    \noindent where $G$ is the ball of radius $U+2T+D$ about $x$, proving the claim.

\end{proof}

\begin{defn} The group of \textit{quantum cellular automata} on a net $A$ is defined to be $\textbf{Net}_{L}(A,A)$. We denote this group $\textbf{QCA}(A)$.
\end{defn}

Quantum cellular automata (QCA) of spin systems have recently been extensively investigated in the physics literature. We will say some words about QCA from a physical viewpoint in the next section. The easiest examples of quantum cellular automata are \textit{finite depth quantum circuit}. Let $A$ be a discrete net of C*-algebras. A depth one quantum circuit in $A$ is a QCA constructed from the following data:

\begin{itemize}

\item
$\{F_{i}\}_{i\in I}$ is a partition of $L$ by finite sets with uniformly bounded diameters. 
\item $\{u_{i}\in A_{F_i}\}$ is a choice of unitaries. 

\end{itemize}

\noindent From this data, we define an automorphism of the quasi-local algebra A. Identify $I$ with the natural numbers $\mathbbm{N}$ (which is possible since we assumed that $L$ is countably infinite), and define $G_{n}=\cup^{n}_{i=1}F_{n}$. Set $v_{n}:=\prod^{n}_{i=1}u_{i}\in A_{G_{n}}$. Then consider $\alpha_{n}:=\text{Ad}(v_{n})\in \text{Aut}(A)$.
For any finite subset $F$, let $n_{0}$ be the smallest natural such that $F\subseteq G_{n_{0}}$. Then, for every $n\ge n_{0}$, if $x\in A_{F}$ we have $\alpha_{n}(x)=\alpha_{n_{0}}(x)$. We define $\alpha_{v}(x):=\lim_{n} \alpha_{n}(x)$, which stabilizes pointwise, and thus gives a $*$-automorphism on the union of local algebras. Since there is a unique C*-norm on any increasing union of finite dimensional C*-algebras, this extends to a $*$-automorphism of the quasi-local algebra. We call automorphisms constructed in this way \textit{depth one} quantum circuits.

In practice, we can simply write
$$\alpha(x):=\left(\ \prod_{i\in I} v_{n}\right) x \left(\ \prod_{i\in I} v^{*}_{n}\right),$$

\noindent which makes sense for any local operator $x\in A_{F}$, since all but finitely many of the $v_{n}$ will commute with $x$. Also note that the spread of a depth one circuit is bounded by the largest diameter of a set in the underlying partition.

\begin{defn}
An automorphism $\alpha\in \textbf{QCA}(A)$ is called a \textit{finite depth quantum circuit} if $\alpha=\alpha_{1}\circ \alpha_{2}\dots \circ \alpha_{n}$ where each $\alpha_{i}$ is a depth one circuit. We denote the set of finite depth circuits $\textbf{FDQC}(A)$.
\end{defn}

\begin{prop}
If $\alpha\in \textbf{Net}_{L}(A,B)$ and $\beta\in \textbf{FDQC}(A)$, then $\alpha\circ \beta \circ \alpha^{-1}\in \textbf{FDQC}(B)$.
\end{prop}

\begin{proof}
Let $\beta\in \textbf{FDQC}(A)$ be depth one, and $\alpha\in \textbf{Net}_{L}(A,B)$ with spread at most $R$. Let $F=\{F_{i}\in \mathcal{F}(L)\}_{i\in I}$ be a collection of finite sets corresponding to $\beta$ and $T\ge 0$ such that $\text{diam}(F_{i})\le T$. Let $u_{i}\in A_{F_i}$ the corresponding unitaries implementing $\beta$.

Consider the graph $G$ with vertex set $I$, defined by declaring $i$ adjacent to $j$ if $N_{3R}(F_{i})\cap F_{j}\ne \varnothing$. This relation is symmetric. Clearly the degree of each vertex is finite. We claim that in addition, the degree is uniformly bounded. Indeed, since each $N_{3R}(F_{i})$ is contained in a ball of radius $T+3R$ of any point in $F_{i}$, by the bounded geometry assumption there exists an $S$ depending only on $T+3R$ such that $|N_{3R}(F_{i})|\le S$ for all i. Therefore, the number of distinct $j$ such that $N_{3R}(F_{i})\cap F_{j}\ne \varnothing$ is bounded by $S$. In particular, the degree of $G$ is uniformly bounded by $S$.

We claim thsat there is a vertex coloring with a finite number of colors. Indeed, for every finite subgraph $G^{\prime}\subseteq G$, the degree is also bounded by $S$, so utilizing the greedy coloring algorithm, we can color $G^{\prime}$ with $S+1$ colors. By the De Bruijn–Erdős theorem \cite{MR0046630}, this implies G can be colored S+1 colors. Choose such a coloring. 

For each color $a\in \{ 1, 2, \dots, S+1\}$, define $I_{a}$ as the set of vertices colored $a$. Consider the family $G^{a}=\{ G_{i}:=N_{R}(F_{i})\in F\ : i\in I_{a}\}$. We can extend this trivially to a partition by adding singletons.
Note that since adjacent vertices have different colors, it is clearly the case that $G_{i}\cap G_{j}=\varnothing$ for any $G_{i},G_{j}\in G^{a}$. Hence the elements of each family are pairwise disjoint. For $i\in I_{a}$, define $w_{i}:=\alpha(u_{i})\in B_{G_i}$ (or $w_i=1$ for the added singletons) and let $\beta_{a}$ denote the corresponding depth one circuit. Note that since $u_{i}$ commutes with $u_{j}$, then $\alpha(u_{i})$ commutes with $\alpha(u_{j})$. Then we see for any local operator $x\in B_{F}$

\begin{align*}
\alpha \circ \beta\circ \alpha^{-1}(x)&=\alpha\left(\left(\prod u_{i}\right)\alpha^{-1}(x)\left(\prod u^{*}_{i}\right)\right)\\
&=\left(\prod \alpha(u_{i})\right) x \left(\prod \alpha(u_{i})^{*}\right)\\
&=\left(\prod_{i_1\in I_{1}} w_{i_1}\right)...\left(\prod_{i_{S+1}\in I_{S+1}} w_{i_{S+1}}\right) x \left(\prod_{i_{S+1}\in I_{S+1}} w^{*}_{i_{S+1}}\right)... \left(\prod_{i_{1}\in I_{1}} w^{*}_{i_{1}}\right)\\
&=\beta_{1}\circ \dots \circ \beta_{S+1}(x)
\end{align*}

\end{proof}

The above proposition shows that $\text{FDQC}$ behaves like a normal subgroup of the groupoid $\textbf{Net}_{L}$ (and, in particular, \textit{ is} a normal subgroup of the automorphism group of any object). We define the equivalence relation $\sim_{FDQC}$ on $\textbf{Net}_{L}(A,B)$ by $\alpha\sim_{FDQC} \beta$ if $\beta^{-1}\alpha\in \textbf{FDQC}(A)$, or equivalently, if $\alpha \beta^{-1}\in \textbf{FDQC}(B)$. By the previous lemma, composition gives a well-defined associative operation on equivalence classes. This leads to the following definition. 

\begin{defn} $\textbf{Net}_{L}/\textbf{FDQC}$ is the groupoid whose

\begin{itemize}
    \item 
    Objects are nets of C*-algebras over $L$.
    \item
    Morphisms are $\textbf{Net}_{L}(A,B)/\sim_{FDQC}$. 
    \item
    Composition is induced from $\textbf{Net}_{L}(A,B)$.
\end{itemize}

\end{defn}

If we have a groupoid homomorphism from $\textbf{Net}_{L}$ which contains $\textbf{FDQC}$ in the kernel of all the automorphism groups of all the objects, then this descends to a well-defined groupoid morphism from $\textbf{Net}_{L}/\textbf{FDQC}$. This restricts to a homomorphism from $\textbf{QCA}(A)/\textbf{FDQC}(A)$ for any $A$.

\begin{rem}
It would be interesting to define a version of $\textbf{FDQC}(A)$ where the elements are the actual sequence of unitaries rather than the resulting automorphisms. One could imagine defining a unitary \textit{2-group} which could be characterized by an anomaly $[\omega]\in H^{3}(\textbf{QCA}(A)/\textbf{FDQC}(A),\ \textit{U}(1))$ in the sense of \cite{2011.13898}.
\end{rem}

\subsection{Physical interpretation of QCA}\label{physical interpretation}

In this subsection, we will discuss two physical interpretations of the group of QCA and the group QCA/FDQC. These correspond to (at least) two natural ways to view QCA of ordinary spin systems from a physical perspective.

The first arises from viewing the structure of a discrete net as a host for the moduli space of local (symmetric) Hamiltonians. In particular, any local Hamiltonian is built from terms living in finite regions with globally bounded diameter. Thus an isomorphism between two nets $\alpha:A\rightarrow B$ which has bounded spread maps local Hamiltonians to local Hamiltonians. In particular, the group $\textbf{QCA}(A)$ can be viewed as the group of symmetries of the moduli space of local Hamiltonians, which have the potential to implement ``dualities" between a priori very different looking Hamiltonians \cite{Aasen_2016, https://doi.org/10.48550/arxiv.2008.08598,https://doi.org/10.48550/arxiv.2112.09091, https://doi.org/10.48550/arxiv.2211.03777, eck2023xxz}. This point of view is particularly interesting in the context of symmetric nets. In this case, QCA are symmetries of the space of local \textit{symmetric} Hamiltonians, and may implement equivalence between symmetric Hamiltonians that have no non-symmetric counterpart, i.e. the symmetric QCA cannot be extended to an ordinary QCA without sacrificing invertibility. We can use QCA to define a natural equivalence relation on local Hamiltonians. We declare two Hamiltonians equivalent if they are in the same orbit under the action of QCA. Since states in the thermodynamic limit of a spin system are just states on the quasi-local algebra, QCA can also be used to define equivalence relations \textit{directly on states themselves} without reference to a parent Hamiltonian.

A second perspective is to view a QCA as a discrete-time unitary dynamics \cite{https://doi.org/10.48550/arxiv.2205.09141}. This extends the standard viewpoint on classical cellular automata as discrete-time updates of configurations to the quantum setting. This class of evolutions retains physical properties such as local causality and quantum reversability while dispensing with the differential equations and local Hamiltonian generator which usually give rise to these properties. Interest in this perspective emerged alongside the rise of quantum computing, where discrete time evolutions are very natural. We note that QCA themselves are generally not realizable as time evolutions generated by local Hamiltonians unless they are circuits, but can nevertheless approximate arbitrary local Hamiltonian evolutions in a certain sense \cite{https://doi.org/10.48550/arxiv.2211.02086}. This partly justifies the study of ``strictly local" QCA which we consider here, as opposed to more general versions of QCA which have tails that arise from time evolutions of local Hamiltonians.

The role of finite depth quantum circuits in phases of quantum matter was first proposed in \cite{PhysRevB.82.155138}. The authors argue that a natural way to consider two ground states of gapped Hamiltonians equivalent is if they are related by a finite depth circuit, which gives a definition that is independent of a choice of parent Hamiltonian. This equivalence relation gives a possible operational definition for ``topological phase" of ground states of gapped Hamiltonians. This can naturally be extended to an equivalence relation on Hamiltonians themselves, where we declare two local Hamiltonians equivalent if one is conjugate to the other by a finite depth circuits, which we call circuit equivalence. Then it is the group QCA/FDQC which acts by symmetries on the moduli space of circuit equivalence classes of Hamiltonians. From the perspective of discrete-time unitary dynamics, it makes sense to consider QCA/FDQC as the group of \textit{topological phases} of discrete-time unitary dynamics.

From both of these view points, it makes sense to say two QCA are topologically equivalent if they differ by a circuit. This leads to the following question.

\begin{problem}\label{thequestion} For a given discrete net $A$, find topological invariants for $\textbf{QCA}(A)$ and apply them to compute $\textbf{QCA}(A)/\textbf{FDQC}(A)$.
\end{problem}

A complete solution to this problem is given for ordinary spin systems on a 1D lattice \cite{MR2890305}. This index has been extended to higher-dimensional manifolds, with a complete classification given in 2D \cite{MR4103966}. However, it is believed this index is insufficient in higher dimensions. Indeed, in three dimensions, there is intriguing evidence that this group should be related to the Witt group of modular tensor categories, or equivalently, invertible fully extended 3+1 D TQFTs \cite{MR4544190, MR4309221, https://doi.org/10.48550/arxiv.2205.09141, PRXQuantum.3.030326}. In general, it is known that the group $\textbf{QCA}/\textbf{FDQC}$ for ordinary spin systems is abelian  \cite{MR4381173} \footnote{We caution the reader that many of the results beyond 1D use more general notions of equivalence of QCA, in particular stable equivalence (adding ancilla locally) and blending. It is not entirely clear what the right version of these notions is in the symmetric setting, since abstract nets of C*-algebras are less flexible than ordinary spin systems.}.

One of the main results of our paper is that even in 1D, in the symmetric case the group $\textbf{QCA}(A)/\textbf{FDQC}(A)$ of an arbitrary net is generally not abelian. Hence we will need invariants beyond a numerical index theory to classify these groups, which is one motivation for the development of DHR theory for symmetric spin systems.

\section{Discrete DHR Theory}\label{DiscreteDHRTheory}

In this section, we develop a version of the DHR theory of superselection theory suitable for abstract spin systems. We note that the usual DHR theory is based on a distinguished Hilbert space representation (the`` vacuum" or ``ground state" representation) and proceeds to study superselection sectors as other representations which ``look like" the vacuum representation outside any small region. This approach has been useful in the study of topologically ordered spin systems \cite{MR2804555, MR3410706, CNN20, MR4362722, wallick2022algebraic}. However, this approach is heavily state dependent and is not well suited for the study of QCA, which depend only on the quasi-local algebra.

In this section, we introduce a version of DHR theory in which the role of states is replaced by ucp (unital completely positive) maps on the quasi-local algebra, and the role of Hilbert space representations is replaced by bimodules. Physically, we can think of this as a superselection theory of quantum channels, rather than a superselection theory of states. The DHR category we define is then the category of superselection sectors of the identity channel. To motivate this conception, we first heuristically review the connection between states, representations and superselection theory.

In the study of quantum spin systems, we are interested in states in the thermodynamic limit (see \cite{MR1441540, MR3617688}), which are modeled by \textit{states} on the quasi-local algebra $A$. Recall a state on the C*-algebra $A$ is simply a positive linear functional $\phi:A\rightarrow \mathbbm{C}$ such that $\phi(1)=1$. In practice, these often arise as ground states or equilibrium states of a local Hamiltonian, but from the quantum information perspective it is desirable to study these states independently of their origin.

Given a state on $A$, we can build a Hilbert space of local perturbations, sometimes called a ``sector". This is achieved by applying the Gelfand-Naimark-Segal (GNS) construction. We start by representing the state $\phi$ formally as the vector state $\Omega_{\phi}$. We introduce other vectors to this Hilbert space by formally adding \textit{local perturbations} of $\phi$, yielding the vectors space $\{a\Omega_{\phi}\ : a\in A\}$. Intuitively, these are the states accessible from $\phi$ by the application of local operators. The inner product of any two of these is defined to be

$$\langle a\Omega_{\phi} |\ b\Omega_{\phi}\rangle :=\phi(a^{*}b)$$

We quotient out by the null vectors and complete this to obtain a Hilbert space denoted $L^{2}(A,\phi)$. This gives a concrete Hilbert space realization of all local perturbations of $\phi$, which is acted upon by $A$. 

We are thus led to extend the \textit{convex set} of states to the \textit{W*-category} $\textbf{Rep}(A)$, whose objects are Hilbert space representations of $A$, and morphisms are bounded linear operators intertwining the actions. The advantage of this approach is that it allows us to consider all local perturbations of a state globally, as an object in the category $\textbf{Rep}(A)$. Thus macroscopic properties of states, which should be invariant under local perturbations, should be expressible as properties of the corresponding GNS representation, opening the door to applying category theory in the study of quantum many-body systems.

Now we recall the theory of superselection sectors from the perspective of algebraic quantum field theory (see \cite{MR1405610}). Given a state $\phi$ on the quasi-local C*-algebra $A$, a representation $H$ is \textit{localizable with respect to $\phi$} if for any (sufficiently large) ball $F$,

$$H|_{A_{F^{c}}}\approx L^{2}(A,\phi)|_{A_{F^{c}}}$$

Here, $\approx$ denotes quasi-equivalence of representations of a C*-algebra \cite{MR887100}, but morally it is useful to think of ``equivalence"\footnote{indeed, in many cases equivalence is automatically implied}. This condition is often called the \textit{superselection criterion}. We also note that we are using ``balls" here primarily for expository purposes, but this is not essential. For example, in applications to topologically ordered spin systems in 2+1 dimensions, infinite cones are the appropriate regions to use.

We interpret a localizable representation as a sector (or collection of states related by local perturbations) that ``looks like" the vacuum sector outside any small region. In other words, the measurable difference from the ground state representation can be localized in any small (but non-empty) region. By zooming out and squinting our eyes, it is reasonable to consider these objects as \textit{topological point defects} of the state $\phi$. ``Topological'' because the region $F$ of localization can be chosen arbitrarily, and ``point'' because balls of finite radius look like points from infinity. 

We define the category of superselection sectors $\textbf{Rep}_{\phi}(A)$ as the W * category of representations satisfying the superselection criteria. In most applications of superselection theory, one proceeds to make some technical assumptions which allow for the construction of a braided monoidal structure on this category. This plays a crucial role in many aspects of chiral conformal field theory and topologically ordered spin systems. Building these structures is highly nontrivial, and it is the study of the braided monoidal structure that we refer to as DHR theory, after the seminal work of Doplicher, Haag, and Roberts \cite{MR258394, MR297259, MR334742}.

In most manifestations of this story, there is a basic state as a fundamental part of the data: in AQFT it is usually part of the definitions (the vacuum state), and in topologically ordered spin systems it arises as the ground state of a Hamiltonian. Superselection theory is then considered relative to that state.

The idea will be to extend the above discussion by replacing states with quantum channels. Suppose now that we have \textit{two} discrete nets of algebras, $A$ and $B$, over the same metric space $L$. Conceptually we make the following substitutions:

\begin{itemize}
    \item 
    States on A $\mapsto $ ucp maps (i.e. quantum channels) from $A$ to $B$.
    \item
    Representations of $A$  $\mapsto$ $A$-$B$ bimodules (right correspondences).
\end{itemize}

This analogy is well known in the theory of operator algebras. Indeed, this is more than an analogy: if we substitute $B=\mathbbm{C}$, we recover states and Hilbert space representations on the nose. Recall that a ucp map $\phi:A\rightarrow A$ is a completely positive map with $\phi(1)=1$. Thus the state-Hilbert space picture is a special case of the ucp map-bimodule correspondence. In quantum information theory, ucp are called ``quantum channels", being the most general type of operation on a quantum system mapping states to states (by composing). 

Like states, ucp maps have an analogue of the GNS construction obtained by taking local perturbations, but instead of producing a Hilbert space representation of A, they result in a right $A$-$B$ correspondence (which should be viewed as a C*-algebra version of ``bimodule", for a detailed definition see Section \ref{bimodules section}). This works as follows:

Let $\phi:A\rightarrow B$ be a ucp map. We build a vector space, starting with the channel $\phi$, represented by the vector $\Omega_{\phi}$ as in the GNS construction. Then the vector space will consist of local perturbations of this channel. We can perturb by operators from $A$ on the left and operators from $B$ on the right, so that we obtain vectors of the form $\{a\Omega_{\phi} b\ :\ a\in A,\ b\in B\}$.

Then we consider a (right) B-valued inner product 

$$\langle a\Omega_{\phi} b\ |\ c\Omega_{\phi} d\rangle:=b^{*}\phi(a^{*}c)d$$

Modding out by the kernel and completing, we obtain a right $A$-$B$ correspondence, which we call $L^{2}(A-B, \phi)$, directly generalizing the GNS construction. This strongly suggests that the analogue of a Hilbert space representation for quantum channels should be a (right) $A$-$B$ correspondences.

From this perspective, it seems plausible that we should be able to define a \textit{superselection category of a quantum channel} rather than of a single state. Here we have the added advantage that, unlike Hilbert space representations, correspondences naturally have a monoidal product (or more precisely, C*-algebras and right correspondences form a 2-category). Furthermore, on any given quasi-local algebra, there is a canonical quantum channel: the identity map. This should then give a canonical, state-independent superselection category for any net of C*-algebras, which naturally has the structure of a C*-tensor category.

We proceed to give a formal definition of this superselection category for a net $A$. This will consist of ``localizable" bimodules, and will naturally assemble into a C*-tensor category. Since this is fairly close in spirit to the DHR perspective of endomorphisms on the quasi-local algebra, we will call this category $\textbf{DHR}(A)$. First, we give some background definitions on bimodules of C*-algebras in the next section.


\subsection{Bimodules of a C*-algebra}\label{bimodules section}. Let $A$ be a (unital) C*-algebra. A (right) Hilbert A-module consists of a vector space $X$, which is a right $A$ module (algebraically), together with an sesquilinear map $\langle \cdot |\ \cdot \rangle: X\times X\rightarrow A$ (conjugate linear in the first variable, linear in the second) satisfying:

\begin{enumerate}
    \item 
    $\langle x\ |\ ya\rangle=\langle x\ |\ y\rangle a$.
    \item
    $\langle x\ |\ x\rangle\ge 0$, with equality if and only if $x=0$.
    \item
    $\langle x\ |\ y\rangle^{*}=\langle y\ |\ x\rangle$.
    \item
    The norm $||x||:=||\langle x\ |\ x\rangle||^{\frac{1}{2}}$ is complete.
    \end{enumerate}

Given two Hilbert A-modules $X$ and $Y$, an adjointable operator from $X$ to $Y$ is an A-module intertwiner $T:X\rightarrow Y$ such that there exists an A-module intertwiner $T^{*}:Y\rightarrow X$ with $\langle T(x)\ | y\rangle_{Y}=\langle x\ |\ T^{*}(y)\rangle_{X}$. The space of adjointable operators is denoted $\mathcal{L}(X,Y)$. $\mathcal{L}(X,X)$ is a unital C*-algebra.

If $A$ is a C*-algebra, an $A$-$A$ \textit{bimodule} is a Hilbert $A$-module $X$, together with a unital *-homomorphism $A\rightarrow \mathcal{L}(X,X)$. We express this homomorphism as a left action, either with standard left multiplication notation, e.g. $ax$, or with triangles, e.g. $a\triangleright x$. In the literature, what we are calling bimodules are usually called (right) correspondences, and we will use the terms interchangeably.

An intertwiner between bimodules $X$ and $Y$ is an element $f\in \mathcal{L}(X,Y)$ such that $f(ax)=af(x)$ (note that $f\in \mathcal{L}(X,Y)$ already implies $f$ intertwines the right $A$ action). The collection of all bimodules and intertwiners assembles into a C*-category which we call $\textbf{Bim}(A)$.

In fact $\textbf{Bim}(A)$ has the structure of a \textit{C*-tensor category}. Recall that C*-tensor categories are C*-categories (see, e.g. \cite{MR808930}) with a linear monoidal structure such that the $*$ operation is compatible with $\otimes$, and the unitors and associators are unitary isomorphisms. For further details, see \cite{MR4419534, MR1444286} and references therein.

To define the tensor product on $\textbf{Bim}(A)$, we consider the A-valued sesquilinear form $$(X\otimes Y)\times (X\otimes Y)\rightarrow A$$ defined by

$$\langle x_1\otimes y_1 |\ x_2 \otimes y_2\rangle_{X\boxtimes_{A} Y}:=\langle y_{1}\ |\ \langle x_1\ |\ x_2\rangle_{X}\ y_{2}\rangle_{Y}$$

\noindent Taking the quotient by the kernel of this form and then completing gives a new $A$-$A$ bimodule denoted by $X\boxtimes_{A} Y$ or simply $X\boxtimes Y$ if the $A$ subscript is clear from context. We will typically denote the image of the simple tensor $x\otimes y$ inside $X\boxtimes_{A} Y$ by $x\boxtimes y$. Then the left and right actions of $A$ are simply given on simple tensors by 

$$a(x\boxtimes y)b:=ax\boxtimes yb$$

Similarly, if $f: X_{1}\rightarrow X_{2}$, $g:Y_{1}\rightarrow Y_{2}$ are bimodule intertwiners, then 

$$f\boxtimes g: X_{1}\boxtimes Y_{1}\rightarrow X_{2}\boxtimes Y_{2}$$

$$(f\boxtimes g)(x\boxtimes y):=f(x)\boxtimes g(y)$$

\noindent gives a well-defined bimodule intertwiner. The obvious ``move the parentheses map" from $(X\boxtimes Y)\boxtimes Z\cong X\boxtimes (Y\boxtimes Z)$ is a natural bimodule intertwiner and satisfies the pentagon identity. Thus $\textbf{Bim}(A)$ is canonically equipped with the structure of a C*-tensor category. For further on the categorical structure see \cite[Section 2]{MR4419534}.

An important ingredient for us are projective bases for correspondences. In the context of subfactors, these were first introduced by Pimsner and Popa \cite{MR860811}, and for inclusions of C*-algebras and bimodules by Watatani \cite{MR996807} and Kajiwara and Watatani \cite{MR1624182}, to study the Jones index \cite{MR0696688}. From an algebraic perspective, these are straightforward analytic extensions of projective bases for modules of associative algebras, hence for this reason we will call them projective bases here.

\begin{defn} Let $X$ be a right Hilbert A-module. A projective basis is a finite subset $\{b_{i}\}^{n}_{i=1}\subseteq X$ such that for all $x\in X$, 
$$\sum_{i} b_{i}\langle b_{i}\ |\ x\rangle=x.$$

\noindent A bimodule is called \textit{right finite} if there exists a projective basis.
\end{defn}

It is easy to see that a right Hilbert module admits a projective basis if and only if it is finitely generated and projective as an $A$ module (hence the terminology). A bimodule is right finite if and only if it has an \textit{amplimorphism } model. These are built from (not necessarily unital) homomorphisms $\pi:A\rightarrow M_{n}(A)$, with the bimodule $X$ given by $\pi(1) A^{n}$, with left action of $\pi$ and right action diagonal. This correspondence is described, for example, in the $\rm{II}_{1}$ factor context in \cite{MR1049618} or more categorically in \cite[Remark 2.12]{2021arXiv211106378C}. Amplimorphisms are closer to the picture of endomorphisms typically used in AQFT.

The collection of right finite bimodules is a full C*-tensor subcategory of $\textbf{Bim}(A)$, since if $\{b_{i}\}$ and $\{c_{j}\}$ are projective bases for $X,Y$ respectively, then $\{b_{i}\boxtimes c_{j}\}$ is a projective basis for $X\boxtimes Y$. If $X$ has a projective basis$\{b_{i}\}$, then $X$ is the A-linear span of the $\{b_{i}\}$. In particular, if $Y$ is another right Hilbert A-module and $f:X\rightarrow Y$ is a right A-module homomorphism, then $f$ is uniquely determined by its action on basis elements.

\bigskip


\subsection{DHR functor}

Let $A$ be a discrete net over the countable bounded geometry metric space $L$. Recall that for any finite region $F$, $A_{F^{c}}$ is the C*-subalgebra of $A$ generate by all $A_{G}$, where $G\in \mathcal{F}(L)$ and $G\cap F=\varnothing$.

\begin{defn}\label{localizable bimodule}
Let $F\in \mathcal{F}(L)$. We say that a right finite correspondence $X$ is \textit{localizable in $F$} if there exists a projective basis $\{b_{i}\}^{n}_{i=1}$ such that for any $a\in A_{F^{c}}$, for each $i$

$$ab_{i}=b_{i}a.$$
\end{defn}

\begin{defn}
Suppose $A$ is a discrete net. Then we say that a right finite correspondence $X$ is \textit{localizable} if there exists an $R\ge 0$ such that $X$ is localizable in all balls of radius at least $R$. We denote the full C*-tensor subcategory of localizable right finite correspondences in $\textbf{Bim}(A)$ by $\textbf{DHR}(A)$
\end{defn}

For a localizble bimodule, we say that the $R$ in the definition is a \textit{localization radius} of $X$. Since we can replace $R$ with any larger $R$, we can assume, without loss of generality, that the localization radius is a positive integer. 

\medskip

Let $\textbf{C*-TensCat}$ be the groupoid defined as follows:

\begin{itemize}
    \item 
    Objects are C*-tensor categories.
    \item
    Morphisms between C*-tensor categories are unitary equivalences between C*-tensor categories up to unitary monoidal equivalence.
    \item
    Composition is induced from composition of equivalences.
\end{itemize}

\begin{thm}\label{TheoremA}(Theorem \ref{ThmA})
    The assignment $A\mapsto \textbf{DHR}(A)$ extends to a  functor $$\textbf{DHR}: \textbf{Net}_{L}\rightarrow \textbf{C*-TensCat}.$$ The corresponding homomorphism $$\textbf{DHR}:\textbf{QCA}(A)\rightarrow \textbf{Aut}_{\otimes}(\textbf{DHR}(A))$$ contains $\textbf{FDQC}(A)$ in its kernel.
\end{thm}

\begin{proof}
First note that for any isomorphism of C*-algebras $\alpha:A\rightarrow B$, we have a canonical equivalence $\alpha_{*}:\textbf{Bim}(A)\rightarrow \textbf{Bim}(B)$. Here the $A$-$A$ bimodule $X$ is sent to $\alpha_{*}(X)\in \textbf{Bim}(B)$, where $\alpha_{*}(X)=X$ as a Banach space, with $B$-$B$ bimodule structure defined for $a,b\in B$, $x,y\in X$ by

$$a\triangleright_{\alpha} x\triangleleft_{\alpha} b:=\alpha^{-1}(a) x \alpha^{-1}(b)$$

$$\langle x | y\rangle_{\alpha_{*}(X)}:=\alpha(\langle x | y\rangle_{X})$$

\noindent This extends to a $*$-functor by defining, for any $f:X\rightarrow Y$,

$$\alpha_*(f) : \alpha_*(X) \ni x \mapsto f(x) \in \alpha_*(Y)$$

\noindent There is an obvious unitary monoidal structure on $\alpha_{*}$, with tensorator

$$\mu^{\alpha}_{X,Y}: \alpha^{*}(X)\boxtimes_{B} \alpha_{*}(Y)\cong \alpha_{*}(X\boxtimes_{A} Y)$$

$$\mu^{\alpha}_{X,Y}(x\boxtimes_{B} y):= x\boxtimes_{A} y.$$

\noindent Also, it's clear from the definition that $\alpha_{*}\circ \beta_{*}\cong(\alpha\circ \beta)_{*}.$

Now the claim is that if $A$ and $B$ are nets over $L$, $X$ is a localizable bimodule over $A$ with localization radius $R$, and $\alpha\in \textbf{Net}_{L}(A,B)$ such that $\alpha^{-1}$ has spread at most $T$, then $\alpha_{*}(X)$ is localizable in $B$, with localizable radius $R+T$. To see this, suppose that $F$ is a ball of radius greater than $R+T$, and let $\{b_{i}\}$ be a projective basis in $X$ localizing in the corresponding ball of radius $R$. Then since the spread is at most $T$, clearly $\{b_{i}\}$ is a projective basis for $\alpha_{*}(X)$ which is localizing in $F$, proving the claim.

We can define 

$$\textbf{DHR}(\alpha):=\alpha_{*}|_{\textbf{DHR}(A)}.$$

\bigskip

Now, to show the second part of the theorem, it suffices to show that for any depth one circuit $\alpha\in \textbf{Net}_{L}(A,A)=\textbf{QCA}(A)$, $\textbf{DHR}(\alpha)$ is monoidally naturally isomorphic to the identity. Suppose $\textbf{F}=\{F_{i}\}_{i\in J}$ is a partition of $L$ with uniformly bounded diameter $T$, and $u_{i}\in A_{F_{i}}$ a choice of unitaries with $$\displaystyle \alpha(a):=\left(\prod_{i\in J} u_{i} \right) a \left(\prod_{i\in J} u^{*}_{i}\right).$$

For any finite subset $F\subseteq L$, define $X_{F}:=\{x \in X :\ ax=xa\ \text{for all}\ a\in A_{F^{c}}\}$. Note that since $X$ is localizable, the union $\bigcup_{F\ \text{is a ball}} X_{F}\subseteq X$ is dense (in fact, we can take the union over \textit{any} increasing sequence of balls). For any $F\subseteq L$, we set $J_{F}=\{i\in J\ : F_{i}\cap F\ne \varnothing\}$. 

We define the map

$$\eta_{X}: \alpha_{*}(X)\rightarrow X $$

by setting, for any $x\in X_{F}$,

\begin{align*}
\eta_{X}(x)&=\left(\prod_{i} u^{*}_i \right)x  \left( \prod_{i} u_{i}\right) \\
&=\left(\prod_{i\in J_{F}} u^{*}_{i}\right) x \left(\prod_{i\in J_{F}} u_{i}\right)\\
&=\left(\prod_{i\in J_{G}} u^{*}_{i}\right) x \left(\prod_{i\in J_{G}} u_{i}\right)\ \text{for any $F\subseteq G$ finite}.
\end{align*}

\noindent This is clearly a norm isometry on this dense subspace, and thus extends to a uniquely defined linear map. To check that it is a bimodule intertwiner, let $a\in A_{I}$, $b\in A_{K}$, and $x\in X_{M}$. Set $N=I\cup K\cup M$

\begin{align*}&\eta_{X}(a\triangleright_{\alpha} x \triangleleft_{\alpha} b)\\
&=\left(\prod_{i\in J_{N}} u^{*}_{i}\right) \left(\left(\prod_{i\in J_{I}} u_{i}\right) a\left(\prod_{i\in J_{I}} u^{*}_{i}\right) x \left(\prod_{i\in J_{K}} u_{i}\right) b \left(\prod_{i\in J_{K}} u^{*}_{i}\right) \right) \left(\prod_{i\in J_{N}} u_{i}\right)\\
&=\left(\prod_{i\in J_{N}} u^{*}_{i}\right) \left(\left(\prod_{i\in J_{N}} u_{i}\right) a\  \text{Ad}\left(\prod_{i\in J_{N}} u^{*}_{i}\right) (x)\  b \left(\prod_{i\in J_{N}} u^{*}_{i}\right) \right) \left(\prod_{i\in J_{N}} u_{i}\right)\\
&= a\left(\prod_{i\in J_{N}} u^{*}_{i}\right)  x  \left(\prod_{i\in J_{N}} u_{i}\right) b  \\
&=a\ \eta_{X}(x)\ b
\end{align*}

\bigskip

\noindent Note that the adjoint of $\eta_{X}$ is 

$$\eta^{*}_{X}(x)=\left(\prod_{i} u_i \right)x  \left( \prod_{i} u^{*}_{i}\right)=\eta^{-1}_{X}(x)$$

\noindent To see that the family $\eta=\{\eta_{X}\}_{X\in \textbf{DHR}(A)}$ is a monoidal natural transformation, we first check naturality. For any bimodule intertwiner $f:X\rightarrow Y$, note that for any finite set $F$, if $x\in X_{F}$, then $f(x)\in Y_{F}$. Then we compute

\begin{align*}
f(\eta_{X}(x))&= f\left(\left(\prod_{i\in J_{F}} u^{*}_{i}\right) x \left(\prod_{i\in J_{F}} u_{i}\right)\right)\\
&= \left(\prod_{i\in J_{F}} u^{*}_{i}\right) f(x) \left(\prod_{i\in J_{F}} u_{i}\right)\\
&=\eta_{Y}(f(x))\\
&=\eta_{Y}(\alpha_{*}(f)(x))\\
\end{align*}

\noindent In the above computation, we have used the fact that the finite product $\left(\prod_{i\in J_{F}} u_{i}\right)\in A$ so is intertwined by $f$. Finally, for monoidality of $\eta$, let $x\in X_{F}$ and $y\in Y_{G}$. Choose some $H\in \mathcal{B}(L)$ with $F\cup G \subseteq H$. Then

\begin{align*}
\mu^{\alpha}_{X,Y}(\eta_{X}\boxtimes \eta_{Y})(x\boxtimes y)&=\eta_{X}(x)\boxtimes \eta_{Y}(y)\\
&=\left(\prod_{i\in J_{H}} u^{*}_{i}\right) x \left(\prod_{i\in J_{H}} u_{i}\right)\boxtimes \left(\prod_{i\in J_{H}} u^{*}_{i}\right) y \left(\prod_{i\in J_{H}} u_{i}\right)\\
&=\left(\prod_{i\in J_{H}} u^{*}_{i}\right) \left( x \boxtimes  y \right)\left(\prod_{i\in J_{H}} u_{i}\right)\\
&=\eta_{X\boxtimes Y}(\mu^{\alpha}_{X,Y}(x\boxtimes y)).
\end{align*}
\noindent Here we have again used the fact that the finite product $\left(\prod_{i\in J_{H}} u_{i}\right)\in A$ and the tensor product is $A$ middle-linear.

\end{proof}


\subsection{Constructing the braiding}

We now follow the usual DHR recipe to build a braiding on the HDR tensor category. However, without additional assumptions we run into problems: braidings may not exist, or may not be unique. In order to avoid these technicalities, for this paper we restrict our attention to lattices in $\mathbbm{R}^{n}$.

\begin{defn} An $n$-dimensional \textit{lattice} is a uniformly discrete subset $L\subseteq \mathbbm{R}^{n}$ such that there is a $C$ with $d(x,L)< C$ for all $x\in \mathbbm{R}^{n}$. We call $C$ a lattice constant. 
\end{defn}

\noindent

For the rest of the section, we will let $L$ be a lattice in $\mathbbm{R}^{n}$ with lattice constant $C$, and $A$ a discrete net on $L$ satisfying weak algebraic Haag duality with duality constants $R,D$ (Definition \ref{algebraicHaag}).
Set $T_{0}:=2C+2D+2R$. We proceed to construct a braiding on $\textbf{DHR}(A)$. 

First we note an immediate consequence of the definition of weak algebraic Haag duality.

\begin{cor}\label{localization of inner product}
Suppose a net satisfies weak algebraic Haag duality with duality constants $R,D$. If $F\in \mathcal{B}(L)$ is a ball of radius $U\ge R$ about a point $x\in L$, $\{b_{i}\}^{n}_{i=1}$ is any F-localizing basis of a correspondence $X$, and $G$ is any ball of radius at least $U+D$ about $x$, then for any $a\in A_{F}$, $\langle b_{i}\ |\ ab_{j}\rangle\in A_{G}$.
\end{cor}

\begin{proof}
It suffices to show $\langle b_{i}\ |\ ab_{j}\rangle\in Z_{A}(A_{F^{c}})$. But for any $b\in A_{F^{c}}$, we have $ab=ba$ so 

\begin{align*}
\langle b_{i}\ |\ ab_{j}\rangle b&=\langle b_{i}\ |\ ab_{j}b\rangle \\
&=\langle b_{i}\ |\ bab_{j}\rangle\\
&=\langle b^{*}b_{i}\ |\ ab_{j}\rangle\\
&=\langle b_{i}b^{*}\ |\ ab_{j}\rangle\\
&=b \langle b_{i}\ |\ ab_{j}\rangle
\end{align*}
\end{proof}

\begin{lem}\label{defineintertwine} Let $X,Y\in \textbf{DHR}(A)$, and let $(x,y)\in L\times L$ with $d(x,y)> T_{0}+R_{X}+R_{Y}$. Let $F=B_{R_{X}}(x)$ and $G=B_{R_Y}(y)$, and $\{b_{i}\}$ and $\{c_{j}\}$ be $F$ and $G$ localizing bases for $X$ and $Y$ respectively. Then the assignment $\sum b_{i}\boxtimes_{A}c_{j} a_{ij}\mapsto \sum c_{j}\boxtimes_{A} b_{i} a_{ij}$ gives a well-defined unitary (hence adjointable) operator of right Hilbert modules

$$u^{F,G}_{X,Y}: X\boxtimes_{A} Y\rightarrow Y\boxtimes_{A} X.$$

\noindent independent of the choice of $F$ and $G$ localizing bases.

\end{lem}

\begin{proof}
First we check

\begin{align*}
\langle u^{F,G}_{X,Y}(\sum b_i\boxtimes c_{j} a_{ij})\ |\ u^{F,G}_{X,Y}(\sum b_{i}\boxtimes c_{j} a_{ij})\rangle &=\langle \sum c_{j}\boxtimes_{A} b_{i}\ a_{ij} | \sum c_{j}\boxtimes_{A} b_{i} a_{ij}\rangle\\
&=\sum a^{*}_{ij}\langle b_i\ | \langle c_{j} | c_{k}\rangle b_l \rangle a_{lk}\\
&=\sum a^{*}_{ij}\ \langle b_i\ |\ b_{l}\rangle\  \langle c_{j}\ | c_{k}\rangle\ a_{lk}\\
&=\sum    \langle c_{j} a_{ij}\ | \langle b_i\ |\ b_{l}\rangle c_{k} a_{lk} \rangle\\
&=\langle \sum b_i\boxtimes c_{j}\ a_{ij}\ |\ \sum b_{i}\boxtimes c_{j}\ a_{ij}\rangle
\end{align*}

In the above computation, we have used the fact that $F^{\prime}\cap G^{\prime}=\varnothing$, where $F^{\prime}:=B_{R_{X}+D}(x)$ and $G^{\prime}=B_{R_{Y}+D}(y)=0$, together with Corollary \ref{localization of inner product}. In particular, this implies our linear map $u^{F,G}_{X,Y}$ preserves the kernel in the relative tensor product and hence is well-defined and an isometry of right $A$-modules.

Computing the adjoint, we see $(u^{F,G}_{X,Y})^{*}(c_j\boxtimes b_i)=b_i\boxtimes c_j=(u^{F,G}_{X,Y})^{-1}(c_j\boxtimes b_i)$, and thus $u^{F,G}_{X,Y}$ is a unitary. 

Now, suppose $\{b^{\prime}_{i}\},\ \{c^{\prime}_{j}\}$ are alternative choices for $F$ and $G$-localizing bases respectively for $X$ and $Y$. Then we see
\begin{align*}
u^{F,G}_{X,Y}(b^{\prime}_{i}\boxtimes c^{\prime}_{j})&=u^{F,G}_{X,Y}\left(\sum_{l,k} b_{l} \langle b_{l}\ |\ b^{\prime}_{i}\rangle \boxtimes c_{k} \langle c_{k}\ |\ c^{\prime}_{j}\rangle \right)\\
&=u^{F,G}_{X,Y}\left(\sum_{l,k} b_{l}\boxtimes c_{k}\ \langle b_{l}\ |\ b^{\prime}_{i}\rangle \langle c_{k}\ |\ c^{\prime}_{j}\rangle\right)\\
&=\sum_{l,k} c_{k}\boxtimes b_{l}\ \langle b_{l}\ |\ b^{\prime}_{i}\rangle \langle c_{k}\ |\ c^{\prime}_{j}\rangle\\
&=c^{\prime}_{j}\boxtimes b^{\prime}_{i}.\\
\end{align*}

\end{proof}

\begin{rem} We henceforth assume that $R_{X}\ge R$ for all $X\in \textbf{DHR}(A)$, otherwise, we simply replace $R_{X}$ by $\text{max}\{R, R_{X}\}$.
\end{rem}

\begin{cor}\label{biggerballs}
Suppose $U\ge R_{X}$, $V\ge R_{Y}$ and $d(x,y)>U+V+T_{0}$. If $F=B_{R_{X}}(x),\ F^{\prime}=B_{U}(x), G=B_{R_{Y}}(y),\ G^{\prime}=B_{V}(y)$, then $u^{F,G}_{X,Y}=u^{F^{\prime},G^{\prime}}_{X,Y}$.
\end{cor}

\begin{proof}

This follows from the previous lemma since bases localized in $B_{R_{X}}(x)$ are also localized in $B_{U}(x)$ (similarly for $y, Y$ and $V$).
    
\end{proof}

\begin{lem}\label{Independenceofpoints} Let $(x_1,y_1),(x_2,y_2)\in L\times L$ satisfy $d(x_{i},y_{i})>R_{X}+R_{Y}+T_{0}$ (and if $n=1$, $x_{i}<y_{i}$). Let $F=B_{R_X}(x_1),G=B_{R_{Y}}(y_1),F^{\prime}=B_{R_X}(x_2),G^{\prime}=B_{R_{Y}}(y_2)$. Then  $u^{F,G}_{X,Y}=u^{F^{\prime},G^{\prime}}_{X,Y}$.
\end{lem}

\begin{proof}

First suppose that $(x_1,y_1)$ and $(x_2,y_2)$ satisfy the property that there exist balls $H$ and $K$ of radius at least $R_{X}$ and $R_{Y}$ respectively such that the corresponding balls $H^{\prime}$ and $K^{\prime}$ with radii increased by $D$ are disjoint, and $F\cup F^{\prime}\subseteq H$ and $G\cup G^{\prime}\subseteq K$. Let $\{b_{i}\}, \{b^{\prime}_{i}\}, \{c_{i}\}, \{c^{\prime}_{i}\}$ be an $F,F^{\prime},G,G^{\prime}$-localizing bases, respectively. Then 
\begin{align*}
u^{F^{\prime},G^{\prime}}_{X,Y}(b_{i}\boxtimes c_{j})&=u_{F^{\prime},G^{\prime}}\left(\sum_{l,k} b^{\prime}_{l} \langle b^{\prime}_{l}\ |\ b_{i}\rangle \boxtimes c^{\prime}_{k} \langle c^{\prime}_{k}\ |\ c_{j}\rangle \right)\\
&=\sum_{l,k} c^{\prime}_{k}\boxtimes b^{\prime}_{l}\ \langle b^{\prime}_{l}\ |\ b_{i}\rangle \langle c^{\prime}_{k}\ |\ c_{j}\rangle\\
&=c_{j}\boxtimes b_{i}=u^{F,G}_{X,Y}(b_{i}\boxtimes c_{j}),\\
\end{align*}

\noindent where we have used the fact that $\langle b^{\prime}_{l}\ |\ b_{i}\rangle \in A_{H^{\prime}}$ and $\langle c^{\prime}_{k}\ |\ c_{j}\rangle\in A_{K^{\prime}}$.

Now we claim that for any pair $(x_{1}, y_{1})$ and $(x_{2}, y_{2})$ as in the hypothesis of this lemma, there exists a sequence of $(x_1,y_{1})=(x^{\prime}_1,y^{\prime}_{1}), \dots (x^{\prime}_{n},y^{\prime}_{n})=(x_2,y_2)$ with $d(x^{\prime}_{i}, y^{\prime}_{i})>R_{X}+R_{Y}+T_{0}$ and there exist disjoint balls $H_{i}, K_{i}$ whose $D$ extensions $H^{\prime}_{i}$ and $K^{\prime}_{i}$ are disjoint, and with $B_{R_{X}}(x^{\prime}_{i})\cup B_{R_{X}}(x^{\prime}_{i+1})\subseteq H_{i}$ and $B_{R_{Y}}(y^{\prime}_{i})\cup B_{R_{Y}}(y^{\prime}_{i+1})\subseteq K_{i}$. By the above argument, this will prove the claim. But the continuous version of this claim in $\mathbbm{R}^{n}$ is clear, and since our lattice $L$ is $C$-close to any point in $\mathbbm{R}^{n}$, the result follows from our assumption that $d(x,y)> 2C+2D+R_{X}+R_{Y}$.

\end{proof}

\begin{defn}
 For $X, Y\in \textbf{DHR}(A)$, define $u_{X,Y}=u^{F,G}_{X,Y}$ where $F=B_{R_{X}}(x)$, $G=B_{R_{Y}}(y)$ and $d(x,y)>R_{X}+R_{Y}+T_{0}$ (in the 1-dimensional case we assume $x<y$). By the above lemma, this is independent of the choice of $(x,y)$.
\end{defn}

\bigskip

\begin{lem}
    For any $X,Y\in \textbf{DHR}(A)$, $u_{X,Y}$ is a bimodule intertwiner.
\end{lem}

\begin{proof}

Let $a\in A_{F}$, where $F$ is some ball of radius $U\ge R_{X}$ about the point $x$. Choose $y$ sufficiently far away, i.e. $d(x,y)>>T_{0}+U+R_{Y}$ (and if $n=1$, $x<y$). Set $G=B_{R_{Y}}(y)$. Then choose $\{b_{i}\}$ and $\{c_{j}\}$ $F$ and $G$ localizing bases for $X$ and $Y$ respectively. Then by Corollary \ref{localization of inner product}, $\langle b_{j}\ | ab_{i}\rangle\in A_{F^{\prime}}$ where $F^{\prime}=B_{U+D}(x)\subset G^{c}$. Thus

\begin{align*}
u_{X,Y}(ab_{i}\boxtimes c_{k})&=u_{X,Y}(\sum_{j}b_{j}\langle b_{j}\ | ab_{i}\rangle\boxtimes c_{k})\\
&=u_{X,Y}(\sum_{j} b_{j}\boxtimes c_{k} \langle b_{j}\ | ab_{i}\rangle)\\
&=u^{F,G}_{X,Y}(\sum_{j} b_{j}\boxtimes c_{k}) \langle b_{j}\ | ab_{i}\rangle\\
&=\sum_{j} c_{k}\boxtimes b_{j}\langle b_{j}\ | ab_{i}\rangle\\
&=c_{k}\boxtimes ab_{i}\\
&=ac_{k}\boxtimes b_{i}\\
&=au_{X,Y}(b_{i}\boxtimes c_{k}).
\end{align*}

\end{proof}

Recall that a \textit{unitary braiding} on a C*-tensor category is a family of natural isomorphisms $u_{X,Y}: X\boxtimes Y\cong Y\boxtimes X$ satisfying coherences called the hexagon identities (see \cite[Chapter 8]{MR3242743} for an extensive introduction). The next theorem shows that the unitary isomorphisms we have built satisfy the coherences of a braiding.

\begin{thm}(c.f. Theorem \ref{ThmB})
The family $\{u_{X,Y}: X\boxtimes_{A} Y\rightarrow Y\boxtimes_{A} X  \}$ defines a unitary braiding on $\textbf{DHR}(A)$.
\end{thm}

\begin{proof}
First we check naturality of $u_{X,Y}$. Let $f: X\rightarrow X^{\prime}$ and $g:Y\rightarrow Y^{\prime}$. We need to show

$$u_{X^{\prime},Y^{\prime}}\circ (f\boxtimes g)=(g\boxtimes f)\circ u_{X,Y}.$$

\bigskip

\noindent Then pick $(x,y)$ such that $d(x,y)>R_{X}+R_{Y}+R_{X^{\prime}}+R_{Y^{\prime}}+T_{0}$, set $H=B_{R_{X}+R_{X^{\prime}}+D}(x)$ and $K=B_{R_{Y}+R_{Y^{\prime}}+D}(y)$. Note that $H\cap K=\varnothing$ so $A_{H}$ commutes with $A_{K}$.

Let $\{b_{i}\},\{b^{\prime}_{i}\}$ be $B_{R_{X}}(x),\ B_{R_{X^{\prime}}}(x)$-localizing bases for $X$ and $X^{\prime}$ respectively, and $\{c_{j}\},\{c^{\prime}_{j}\}$ $B_{R_{Y}}(y),\ B_{R_{Y}}(y)$ localizing bases for $Y,Y^{\prime}$ respectively. Then $\langle b^{\prime}_{l}\ |\ f(b_{i})\rangle\in A_{H}$ and $\langle c^{\prime}_{k}\ |\ g(c_{j})\rangle\in A_{K}$ by Corollary \ref{localization of inner product}.

It suffices to check naturality for morphisms evaluated on (any) projective basis elements, and we compute

\begin{align*}
u_{X^{\prime},Y^{\prime}}\circ (f\boxtimes g)(b_{i}\boxtimes c_{j})&=\sum_{l,k} u_{X^{\prime},Y^{\prime}}\left(b^{\prime}_{l}\langle b^{\prime}_{l}\ |\ f(b_{i})\rangle \boxtimes c^{\prime}_{k}\langle c^{\prime}_{k}\ |\ g(c_{j})\rangle\right)\\
&=\sum_{l,k} u_{X^{\prime},Y^{\prime}}(b^{\prime}_{l}\boxtimes c^{\prime}_{k}) \langle b^{\prime}_{l}\ |\ f(b_{i})\rangle \langle c^{\prime}_{k}\ |\ g(c_{j})\rangle\\
&=\sum_{l,k} c^{\prime}_{k}\langle c^{\prime}_{k}\ |\ g(c_{j})\rangle\boxtimes b^{\prime}_{l} \langle b^{\prime}_{l}\ |\ f(b_{i})\rangle\\
&=g(c_j)\boxtimes f(b_i)\\
&=(g\boxtimes f)\circ u_{X,Y}(b_{i}\boxtimes c_{j}).
\end{align*}

\noindent

Now we check the hexagon identity. Let $X, Y, Z\in \textbf{DHR}(A)$. Choose points $x,y,z\in L$ with the distance between any two greater than $R_{X}+R_{Y}+R_{Z}+T_{0}$, and such that there is a ball $K$ around $z$ containing $B_{R_{Y}}(y)\cup B_{R_{Z}}(z)\subseteq K$ with $K\cap B_{R_{X}}(x)=\varnothing$.

Then if $\{b_i\}, \{c_{i}\}, \{d_{i}\}$ localize $X,Y,Z$ in $B_{R_{X}}(x), B_{R_{Y}}(y), B_{R_{Z}}(z)$ respectively, we have that $\{c_{j}\boxtimes d_{k}\}$ localizes $Y\boxtimes Z$ in $K$. Denoting $F=B_{R_{X}}(x)$, we have

\begin{align*}
(1_{Y}\boxtimes u_{X,Z})\circ (u_{X,Y}\boxtimes 1_{Z})(b_{i}\boxtimes c_{j}\boxtimes d_{k})&=c_{j}\boxtimes d_{k}\boxtimes b_{i}\\
&=u^{F,K}_{X, Y\boxtimes Z}(b_{i}\boxtimes c_{j}\boxtimes d_{k})\\
&=u_{X,Y\boxtimes Z}(b_{i}\boxtimes c_{j}\boxtimes d_{k}),
\end{align*}

\noindent where the last equality follows from Corollary \ref{biggerballs}.
This gives us one of the hexagon identities. The other follows from a similar argument. In the above computation, we have suppressed the associator, which acts on basis elements $(b_{i}\boxtimes c_{j})\boxtimes d_{k}\mapsto b_{i}\boxtimes (c_{j}\boxtimes d_{k})$.

\end{proof}

\begin{cor}
For a net $A$ over a lattice $L\subseteq \mathbbm{R}^{n}$ with $n\ge 2$, the braiding on $\textbf{DHR}(A)$ is symmetric.
\end{cor}

\begin{proof}
    Any pair of points in $\mathbbm{R}^{n}$ can be connected to each other in the manner of the proof of Lemma \ref{Independenceofpoints}. We see that $u_{X,Y}=u^{F,G}_{X,Y}=u^{G,F}_{X,Y}=(u_{Y,X})^{-1}$, where the last equality follows from the definition of $u^{F,G}_{X,Y}$.
\end{proof}

By the Doplicher-Roberts theorem, any symmetric C*-tensor category with simple unit is equivalent to $\textbf{Rep}(G,z)$ where $(G,z)$ is a supergroup \cite{MR1010160}. In particular, the pair $(G,z)$ is interpreted as the (global) gauge (super)-group of the theory. In general, when we have an abstract net of C*-algebras, we should think of the braided tensor category $\textbf{DHR}(A)$ as the representation category of some generalized symmetry $G$ acting on an ordinary spin system, with $A$ the net of local symmetric operators.

\begin{thm}(c.f. Theorem \ref{ThmB}). If $A,B$ are nets on $L$ satisfying weak algebraic Haag duality, then for any $\alpha\in Net_{L}(A,B)$, the unitary monoidal equivalence $\textbf{DHR}(\alpha): \textbf{DHR}(A)\cong \textbf{DHR}(B)$ is braided.
\end{thm}

\begin{proof}
Let $X,Y\in \textbf{DHR}(A)$. Suppose $\alpha$ has spread at most $S$. Choose balls $F,G$ such that $N_{S}(F)\cap N_{S}(G)=\varnothing$. Then pick $F$ and $G$ localizing bases $\{b_{i}\}, \{c_{j}\}$ respectively, for $X$ and $Y$ respectively. Let $F^{\prime}=N_{S}(F), G^{\prime}=N_{S}(G)$. Then $\{b_{i}\}$ and $\{c_{i}\}$ are $F^{\prime}$ and $G^{\prime}$ localizing bases, respectively, of $\alpha_{*}(X)$ and $\alpha_{*}(Y)$, respectively. Here we are using the notation $\alpha_{*}$ for $\textbf{DHR}(\alpha)$ as in the proof of Theorem \ref{TheoremA}. We compute

\begin{align*} 
(\mu^{\alpha}_{X,Y})^{*}\circ \alpha_{*}(u_{X,Y})\circ \mu^{\alpha}_{X,Y}(b_{i}\boxtimes_{B} c_{j})&= (\mu^{\alpha}_{X,Y})^{*}(u^{F,G}_{X,Y}(b_i \boxtimes_{A} c_{j}))\\
&=(\mu^{\alpha}_{X,Y})^{*}(c_{j}\boxtimes_{A} b_{i})\\
&=c_{j}\boxtimes_{B} b_{i}\\
&=u^{F^{\prime}, G^{\prime}}_{\alpha_{*}(X),\alpha_{*}(Y)}(b_{i}\boxtimes_{B} c_{j})\\
&=u_{\alpha_{*}(X),\alpha_{*}(Y)}(b_{i}\boxtimes_{B} c_{j}).
\end{align*}

Since module maps are determined on projective basis elements, this proves the claim.

\end{proof}

\section{1D spin systems with categorical symmetries}\label{categorical symmetry}

Recall that a \textit{unitary fusion category} is a semisimple C*-tensor category with simple unit, duals, and finitely many isomorphism classes of simple objects. Fusion categories simultaneously generalize finite groups and their representation categories, and have become important tools for understanding generalized symmetries in mathematics and physics \cite{MR2183279, MR3242743}. Recently, there has been significant interest in fusion categorical symmetries on spin chains, part of a larger interest in non-invertible symmetry \cite{freed2022topological}. One motivation is the search for exotic conformal field theories \cite{PhysRevLett.128.231602, PhysRevLett.128.231603}.

\medskip

There are (at least) two equivalent pictures to describe categorical symmetries:

\medskip

\begin{enumerate}
\item
The first way is to have fusion categories act by \textit{matrix product operators} (MPOs) \cite{Charact, MR3614057,GarreRubio2023classifyingphases,MR3719546,MR4272039,MR4109480}. Mathematically, the data that characterize this are described by a module category $\mathcal{M}$ for $\cC$, and an object $X\in \mathcal{C}_{*}(\mathcal{M})$ in the dual category (\cite{10.21468/SciPostPhys.10.3.053, GarreRubio2023classifyingphases}). The operators localized on n-sites invariant under this symmetry are isomorphic to $\text{End}_{\mathcal{C}_{*}(\mathcal{M})}(X^{\otimes n})$.

\item 
Equivalently, we can consider a weak C*-Hopf algebra $H$ (\cite{MR1726707}) acting on a physical on-site Hilbert space $K$ of spins \cite{WeakHopI, Molnar:2022nmh, MR1463825}, which is a more straightforward generalization of on-site group symmetry. Then $K\in \text{Rep}(H)$, and we can consider the $n$-site Hilbert space $K^{\boxtimes n}$, which is equipped with an action of $H$ using the coproduct. We note that $K^{\boxtimes n}\subseteq K^{\otimes n}$ but if $H$ is not a Hopf algebra, these are not equal. There is a distinguished subalgebra $S\le H$, and any module $K$ becomes a bimodule over $S$. Then $K^{\boxtimes n}\cong K\otimes_{S} K\otimes _{S}\dots K$. The local observable are given by the $H$ intertwining endomorphisms, $\text{End}_{H}(K^{\boxtimes n})$.

\end{enumerate}

In both of these situations, the resulting nets of algebras are described by abstract nets of algebras built directly in terms of abstract fusion categories. This allows us to analyze the theory without worrying about the physical realization of the original spin system. This will also cover the example \ref{BoundaryEx}, which we will discuss in detail in the sequel. For any unitary fusion category $\mathcal{D}$ (which we assume is strict for convenience) and any object $X\in \mathcal{D}$, we can define a net of finite dimensional C*-algebras on the lattice $\mathbbm{Z}\subseteq \mathbbm{R}$. For any interval $I$ with $n$-sites, we set 

$$A_{I}:=\mathcal{D}(X^{n},X^{n}).$$ 

\noindent Here we use the notation $\mathcal{D}(X,Y)$ as shorthand for the morphism space $\text{Hom}(X,Y)$ in the category $\mathcal{D}$, and $X^{n}$ is shorthand for the $n$-fold tensor power $X^{\otimes n}$ in $\mathcal{D}$.

Now suppose $I=[a,b]$ and $J=[c,d]$ with $I\subseteq J$ (so $c\le a$ and $b\le d$). Then we can define the inclusion $A_{I}\subseteq A_{J}$ by identifying 

$$f\mapsto 1_{X^{a-c}}\otimes f \otimes 1_{X^{d-b}},$$

\noindent We then take the colimit over the directed set of intervals in the category of C*-algebras to obtain the quasi-local algebra 

$$A:=\varinjlim A_{I}.$$

For any interval, we denote the inclusion $i_{a,b}: A_{[a,b]}\hookrightarrow A$, and identify $A_{[a,b]}$ with it's image.

\begin{prop}\label{fusionspinchain}
The assignment $F\mapsto A_{F}$ constructed above defines a discrete net of C*-algebras over $\mathbbm{Z}\subseteq \mathbbm{R}$. We call the nets constructed this way \textit{fusion spin chains}.
\end{prop}

The goal of this section is to characterize the DHR category of a fusion spin chain. We will see that under some mild assumptions on the tensor generator, the DHR category of a fusion spin chain is equivalent to the Drinfeld center of the underlying fusion category. To prove this, we will apply ideas and results from subfactor theory. The next section is mostly expository, and will include a summary of the machinery we will utilize to obtain the main results of this section.

\subsection{Quantum symmetries: definitions and results}\label{QuantumSymms}

We will now review some concepts and results originating from the theory of subfactors. We discuss their modern manifestation in terms of fusion category theory and their recent extension from the W*-setting to the C*-setting.

\subsubsection{Actions of fusion categories on C*-algebras}\label{actionssection}

\begin{defn} If $\mathcal{C}$ is a unitary fusion category and $A$ is a (unital) C*-algebra, an \text{action} of $\mathcal{C}$ on $A$ is a C*-tensor functor $F:\mathcal{C}\rightarrow \textbf{Bim}(A)$. 
\end{defn}

Unpacking this definition slightly further, the data of an action of $\mathcal{C}$ on $A$ is: for every object $X\in \mathcal{C}$ a bimodule $F(X)\in \textbf{Bim}(A)$, for $f\in \mathcal{C}(X,Y)$ a bimodule intertwiner $F(f): F(X)\rightarrow F(Y)$, and for every pair of objects $X,Y\in \mathcal{C}$ a unitary isomorphism $F^{2}_{X,Y}:F(X)\boxtimes_{A} F(Y)\rightarrow F(X\otimes Y)$. This data is required to satisfy coherences: $F$ should be $*$-functor, $F^{2}_{X,Y}$ should be natural in $X$ and $Y$, and the family $\{F^{2}_{X,Y}\}$ should satisfy associativity constraint with respect to the bimodule associator (see \cite[Chaper 2.4]{MR3242743}). 

Actions of unitary fusion categories on finite dimensional C*-algebras are well understood in terms of module categories for $\mathcal{C}$ (for example, see Corollary 3.6 in \cite{https://doi.org/10.48550/arxiv.2207.11854} and the discussion therein). Note, however, that any action of a non-trivial fusion category on a finite dimensional algebra is never fully faithful. An AF-action is an action on an AF C*-algebra built out of these finite dimensional pieces, and these can be fully faithful. AF-actions are the actions that are relevant for analyzing the DHR category of fusion spin chains. To give a proper account of AF-actions, we include the following definition for sake of completeness:

\begin{defn}\label{cocyclemorphism}\cite[c.f. Lemma 3.8]{https://doi.org/10.48550/arxiv.2207.11854} Let $A,B$ be unital C*-algebras and $\phi:A\rightarrow B$ a unital $*$-homomorphism. Let $\mathcal{C}$ be a unitary fusion category and suppose we have actions $F:\mathcal{C}\rightarrow \textbf{Bim}(A)$ and $G:\mathcal{C}\rightarrow \textbf{Bim}(B)$. An equivariant structure on $\phi$ with respect to $F$ and $G$ is a family of linear maps $\{k^{X}: F(X)\rightarrow G(X):\ X\in \mathcal{C}\}$ satisfying the following conditions

\begin{enumerate}
\item 
For $a,b\in A$, a$x\in X$, 

$k^{X}(a\triangleright x\triangleleft b)=\phi(a)\triangleright k^{X}(x)\triangleleft \phi(b).$
\item 
For $f\in \mathcal{C}(X,Y)$, $x\in X$
$k^{Y}\circ F(f)(x)=G(f)\circ k^{X}(x).$
\item
$\langle k^{X}(x)\ |\ k^{X}(y)\rangle_{B}=\phi(\langle x\ |\ y\rangle_{A}).$
\item
$G(X)=k^{X}(F(X))\triangleleft B$.
\item 
The following diagram commutes:

\begin{tikzcd}
F(X)\otimes F(Y)
\arrow[swap]{d}{k^{X}\otimes k^{Y}}
\arrow{r}{}
& F(c)\boxtimes_{A} F(d)
\arrow{r}{F^{2}_{X,Y}}
& F(X\otimes Y)\arrow{d}{k^{X\otimes Y}}
 \\
G(X)\otimes G(Y)
\arrow{r}{}
& G(X)\boxtimes_{B} G(Y)
\arrow{r}{G^{2}_{X,Y}}
& G(X\otimes Y)
\end{tikzcd}

\end{enumerate}
\end{defn}

Let $\mathcal{C}$ be a (strict) unitary fusion category. Suppose that we have:

\begin{enumerate}
    \item 
A sequence of finite dimensional C*-algebras $A_{n}$ and unital, injective $*$-inclusions $\iota_{n}: A_{n}\rightarrow A_{n+1}$.
\item 
A sequence of actions $F_{n}:\mathcal{C}\rightarrow \textbf{Bim}(A_{n})$.
\item 
A family $\{k^{X}_{n}\}$ of equivariant structures on $\iota_{n}$ with respect to $F_{n}$ and $F_{n+1}$.
\end{enumerate}

Then if we let $A:=\varinjlim A_{n}$ be the inductive limit of the sequence $A_{n}$ in the category of C*-algebras, there exists a canonical action $F:\mathcal{C}\rightarrow \textbf{Bim}(A)$ called the inductive limit action of the $F_{n}$ (for a detailed construction, see \cite[Proposition 4.4]{https://doi.org/10.48550/arxiv.2207.11854}). Any action of $\mathcal{C}$ on an AF C*-algebra equivalent (in the sense of \cite[Definition 3.9]{https://doi.org/10.48550/arxiv.2207.11854}) to one constructed as above is called an \textit{AF-action}.

Before we go into our main examples, we recall the following definition.

\begin{defn} A self-dual object $X\in \mathcal{D}$ is called \textit{strongly tensor generating} if there exists some $n$ such that every simple object $Y$ is a summand of $X^{ n}$.
\end{defn}

The canonical example of a strong tensor generator is simply the direct sum over all simple objects with multiplicity $1$. For any tensor generator $X$, the object $X\oplus \mathbbm{1}$ will be strongly tensor generating. The self-duality condition we use in the definition is not strictly necessary and implies a kind of spatial reflection symmetry on the fusion spin chain built from $X$. For us it is a matter of convenience, since we can use this assumption to compare our C*-algebra constructions with subfactor theory. In particular, it allows us to use Ocneanu compactness directly (see Remark \ref{Fully faithful}).

\begin{rem}\label{FrobRem}
    If $X^{n}$ contains all isomorphism classes of simple objects, then so does $X^{m}$ for any $m\ge n$. More generally, suppose $Y$ is some object such that every isomorphism class of simple object appears as a summand. Then for any object $Z$, $Y\otimes Z$ also satisfies this property. Indeed, by semisimplicity, for a simple $W$, $W$ appears as a summand of $Y\otimes Z$ if and only if $\mathcal{C}(W, Y\otimes Z)\ne 0$. But by Frobenius reciprocity, $\mathcal{C}(W, Y\otimes Z)\cong \mathcal{C}(W\otimes \overline{Z}, Y)$, where $\overline{Z}$ denotes the dual object. But the latter space is non-zero since $Y$ contains a copy of all simple objects (up to isomorphism).
\end{rem}

\begin{ex}\label{standardactions} \textbf{Standard AF-actions}. We recall the ``standard AF-actions" of fusion categories that have historically played an important role in subfactor theory.


Now, given a strong tensor generator $X$, we construct an AF-action as follows. First, by replacing $X$ with a sufficiently large tensor power, we can assume $X$ itself contains every isomorphism class of simple object, which by Remark \ref{FrobRem} implies every power of $X$ will as well.

Set $A_{n}:=\mathcal{C}(X^{n}, X^{n})$. This is a finite dimensional C*-algebra, whose matrix summands are indexed by isomorphism classes of simple objects. There is a natural unital inclusion $\iota_{n}:A_{n}\rightarrow A_{n+1}$ given by 

$$a\mapsto 1_{X}\otimes a.$$

The inductive limit algebra $A:=\lim A_{n}$ is an AF C*-algebra. $A$ is a simple AF-algebra with a unique tracial state since it has a simple stationary Bratteli diagram (see, for example, \cite{MR312282} and \cite[Chapter 6]{MR623762}). We will now build an AF-action of $\mathcal{C}$ on $A$.

For any $Y\in \mathcal{C}$, set $F_{n}(Y):=\mathcal{C}(X^{n}, X^{n}\otimes Y)$. This has the structure of an $A_{n}$-bimodule with $$a\triangleright \xi \triangleleft y:= (x\otimes 1_{Y})\circ \xi \circ b,$$

for $a,b\in A_{n}, \xi\in F_{n}(Y)$. The right $A_{n}$-valued inner product is

$$\langle \xi\ |\ \eta\rangle_{A_{n}}:=\xi^{*}\circ \eta.$$

For $f\in \mathcal{C}(Y,Z)$, $F_{n}(f)(\xi):=(1_{X^{n}}\otimes f)\circ \xi$, which is clearly a bimodule intertwiner. It is straightforward to check that $F_{n}:\mathcal{C}\rightarrow \textbf{Bim}(A_{n})$ is a C*-functor.

The monoidal structure $(F_{n})^{2}_{Y,Z}: F_{n}(Y)\boxtimes_{A_n} F_{n}(Z)\rightarrow F_{n}(Y\otimes Z)$ is induced by the linear map $$(F_{n})^{2}_{Y,Z}(\xi\otimes \eta):=(\xi\otimes 1_{Z})\circ \eta\in \mathcal{C}(X^{n}, X^{n}\otimes Y\otimes Z)=F_{n}(Y\otimes Z).$$

\noindent It is easy to check that these extend to natural unitary isomorphisms satisfying the required associativity constraints (it is here where we use the strong tensor generator assumption, i.e. that all simple objects appear as summands of all tensor powers of $X$).

Now, we define $k^{Y}_{n}: F_{n}(Y)\rightarrow F_{n+1}(Y)$ by

$$k^{Y}_{n}(\xi):=1_{X}\otimes \xi.$$

\noindent It is straightforward to verify that this defines an equivariant structure on $\iota_{n}$ with respect to $F_{n}$ and $F_{n+1}$. Taking the limit, we obtain a AF-action $F:\mathcal{C}\rightarrow \textbf{Bim}(A)$ which we call a \textit{standard AF-action}.

\end{ex}

\begin{rem}\label{Fully faithful} Standard actions have the nice property of being \textit{fully faithful}, namely for any objects $Y,Z\in \mathcal{C}$, $F: \mathcal{C}(Y,Z)\rightarrow \textbf{Bim}(A)(F(Y),F(Z))$ is an isomorphism. Indeed, this follows from a standard application of Ocneanu compactness \cite{MR1055708},\cite[Chapter 5]{MR1473221} to the subfactor $N\subseteq M$, where $M$ is the $\rm{II}_{1}$ factor obtained from completing $A$ in the GNS representation of its unique trace, and $N$ is the completion of the ``shifted" subalgebra $1_{X}\otimes A\subseteq A$ (see \cite[Theorem 5.1 and Section 6.1]{https://doi.org/10.48550/arxiv.2010.01067}. This is the only place where we will need self-duality of the tensor generator $X$, so that we can directly apply Ocneanu compactness. If $X$ were not self-dual, our tower of algebras would not be a standard $\lambda$-lattice (in the sense of \cite{MR1334479}) since it would lack Jones projections. In this case we would not be able to apply the theorems of subfactor theory directly, and instead would need to appy a more general version of Ocneanu compactness (for example, see \cite{doi:10.1142/S1793525323500589}).
\end{rem}

\subsubsection{Module categories and Q-systems}\label{modulecatandQsys}

If $\mathcal{C}$ is a unitary fusion category, recall that a unitary module category $\mathcal{M}$ is a finitely semisimple C*-category, together with a C*-bifunctor $\mathcal{C}\times \mathcal{M}\rightarrow \mathcal{M}$ and a coherent natural associator (see \cite[Chapter 7]{MR3242743} for definitions). By MacLane's coherence theorem, without loss of generality we can assume that our module category has trivial associator (i.e. is strict). To set some notation, for $Y\in \mathcal{C}$, $m\in \mathcal{M}$ we denote the image of the bifunctor by $Y\triangleright m$, and for $f\in \mathcal{C}(Y,Z)$, $g\in \mathcal{M}(m,n)$, we denote the image under the functor by $f\triangleright g\in \mathcal{M}(Y\triangleright m, Z\triangleright n)$. Strictness of the module category is expressed in this notation by 

$$Y\triangleright (Z\triangleright m)=(Y\otimes Z)\triangleright m$$

and

$$f\triangleright (g\triangleright h)=(f\otimes g)\triangleright h.$$

Associated to a $\mathcal{C}$-module category $\mathcal{M}$ is the unitary multifusion category of $\mathcal{C}$-module endofunctors $\text{End}_{\mathcal{C}}(\mathcal{M})$ (see \cite[Chapter 7]{MR3242743} or \cite{MR3933035}). This is fusion precisely when $\mathcal{M}$ is indecomposable. In this case, we define the dual fusion category $\mathcal{C}^{*}_{\mathcal{M}}:=\text{End}_{\mathcal{C}}(\mathcal{M})^{mp}$ where the superscript mp denotes the monoidal opposite category.

\begin{rem}
For a unitary fusion category $\mathcal{D}$, we recall that $\mathcal{D}^{mp}$ denotes the monoidal opposite category. The objects of $\mathcal{D}^{mp}$ are the same as the objects of $\mathcal{D}$, but we denote the version in $\mathcal{D}^{mp}$ with an mp superscript. Then $\mathcal{D}^{mp}(X^{mp}, Y^{mp}):=\mathcal{D}(X,Y)$, and the composition, $*$-structure and norm are the same as in $\mathcal{D}$. The difference from $\mathcal{D}$ is the monoidal product. The monoidal product is given by $X^{mp}\otimes Y^{mp}:=(Y\otimes X)^{mp}$, with the obvious extension to morphisms and the choice of associator.
\end{rem}

Let $\mathcal{M}$ be an indecomposable module category over the unitary fusion category $\mathcal{C}$. If we pick any $m\in \mathcal{M}$, then we can take the \textit{internal end} construction to obtain an algebra object $\underline{\text{End}}(m)\in \mathcal{C}$, called the \textit{internal endomorphism} of the object $m$. This algebra object is techincally only defined up to isomorphism, but there is a choice of representation which is a \textit{Q-system}. A Q-system is a unital, associative algebra object $Q$ such that the adjoint of the multiplication map is a right inverse for multiplication (or in other words, is an isometry) as well as a $Q$-$Q$ bimodule intertwiner from $Q$ to $Q\otimes Q$. Q-systems are C*-Frobenius algebra object in $\mathcal{C}$ (for detailed definitions and discussions, we refer the reader to the comprehensive references \cite{MR3308880,MR4419534,doi:10.1063/5.0071215,MR3933035}). We will describe the internal endorphism as an associative algebra object following \cite{MR3687214}, and refer the readers to the above-mentioned references for Q-system details. As an object,

$$\underline{\text{End}}(m)\cong \bigoplus_{Y\in \text{Irr}(\mathcal{C})} \mathcal{M}(Y\triangleright m, m)\otimes Y.$$

Here, if $V$ is a finite dimensional Hilbert space\footnote{We have intentionally not yet specified an inner product on $\mathcal{M}(Y\triangleright m, m)$, see Remark \ref{Qsystemsubtelty}} and $Y$ is a simple object in $\mathcal{C}$, $V\otimes Y$ represents the object $Y^{\oplus \text{dim}(V)}$, where we explicitly identify the multiplicity space $\mathcal{C}(Y, V\otimes Y)$ with $V$. $\text{Irr}(\mathcal{C})$ represents a fixed set of representatives of isomorphism classes of simple objects. Using this notation, for any (not necessarily simple) object $Z$ we have $$ Z\cong \bigoplus_{Y\in \text{Irr}(\mathcal{C})} \mathcal{C}(Y, Z)\otimes Y$$ 
\noindent by semisimplicity, where $\mathcal{C}(Y, Z)$ is equipped with the composition inner product, $$\langle f \ |\  g\rangle 1_{Y} =f^{*}\circ g.$$ In particular, for the object $Z:=V\otimes Y$, $V$ is identified with $\mathcal{C}(Y,Z)$. Thus 
$$\mathcal{C}(Y,\underline{\text{End}}(m))\cong \mathcal{M}(Y\triangleright m, m)$$ 

\noindent for any simple $Y$. We call this notation and perspective on expressing objects the \textit{Yoneda representation}. In the Yoneda representation,

\begin{align*}
Z_{1}\otimes Z_{2}&\cong \bigoplus_{Y\in \text{Irr}(\mathcal{C})} \mathcal{C}(Y, Z_{1}\otimes Z_{2})\otimes Y\\
&\cong \bigoplus _{Y,U,W} \left( \mathcal{C}(U,Z_{1})\otimes \mathcal{C}(Y,U\otimes W)\otimes \mathcal{C}(W,Z_{2})\right) \otimes Y
\end{align*}

\noindent One advantage of the Yoneda representation is that it makes morphisms in the category $\mathcal{C}$ expressible purely in terms of (ordinary) linear transformations. Indeed, if $W\cong \bigoplus_{Y\in \text{Irr}(\mathcal{C})} \mathcal{C}(Y, W)\otimes Y$, then by semisimplicity any morphism $f\in \mathcal{C}(Z,W)$ is uniquely determined by a family of linear transformations 

$$f\sim \{f_{Y}: \mathcal{C}(Y,Z)\rightarrow \mathcal{C}(Y,W)\ |\ Y\in \text{Irr}(\mathcal{C})\}.$$

\noindent Given and $f\in \mathcal{C}(Z,W)$, and $g\in \mathcal{C}(Y,Z)$, $f_{Y}(g):=f\circ g\in \mathcal{C}(Y,W)$. That fact that the correspondence $f\sim \{f_{Y}\}_{Y\in \text{Irr}(\mathcal{C})}$ uniquely determines $f$ follows immediately from the Yoneda lemma, hence the origin of the terminology for this picture.

Now, returning to internal end objects to define a multiplication morphism $\mu: \underline{\text{End}}(m)\otimes \underline{\text{End}}(m)\rightarrow \underline{\text{End}}(m)$, observe

$$\underline{\text{End}}(m)\otimes \underline{\text{End}}(m)\cong \bigoplus_{Y,W,Z\in \text{Irr}(\mathcal{C})} \left(\mathcal{M}(Y\triangleright m, m)\otimes \mathcal{C}(W,Y\otimes Z)\otimes \mathcal{M}(Z\triangleright m, m)\right)\otimes W$$

\noindent Thus for $W\in \text{Irr}(\mathcal{C})$, we define $$\mu_{W}: \bigoplus_{Y,Z\in \text{Irr}(\mathcal{C})}\mathcal{M}(Y\triangleright m, m)\otimes \mathcal{C}(W,Y\otimes Z)\otimes \mathcal{M}(Z\triangleright m, m)\rightarrow \mathcal{M}(W\triangleright m, m)$$

\noindent on homogeneous tensors $\xi\otimes f \otimes \nu\in \mathcal{M}(Y\triangleright m, m)\otimes \mathcal{C}(W,Y\otimes Z)\otimes \mathcal{M}(Z\triangleright m, m)$ by 

$$\mu_{W}(\xi\otimes f \otimes \nu):= \xi \circ (1_{Y}\otimes \nu )\circ (f\triangleright 1_{m}).$$

\noindent Then $\mu:=\{\mu_{W}\}$ equips $\underline{\text{End}}(m)$ with the structure of an associative algebra object, and in fact a Q-system (see Remark below).

\begin{rem}\label{Qsystemsubtelty} There is one subtlety that we are sweeping under the rug in this discussion, namely we have not specified the Hilbert space structures on the $\mathcal{M}(Y\triangleright m, m)$. We need to do this to actually pin down a morphism for the object in $\mathcal{C}$ rather than the version of $\mathcal{C}$ that forgets the dagger structure. For specifying an algebra structure on $\underline{\text{End}}(m)$ this is not relevant, since we can apply the definitions above to obtain isomorphic algebra structures for any choice. However, the definition of Q-system requires constraints on the adjoint of the multiplication map, and this is sensitive to which Hilbert space structures we put on the multiplicity spaces. A choice can always be made making this into a Q-system (see \cite[Theorem 4.6]{https://doi.org/10.48550/arxiv.2207.11854} for this level of generality) that is essentially a choice of a unitary module trace, but an in-depth discussion would take us too far afield. We have chosen to not include this discussion since we are satisfied considering only the algebra structure on $\underline{\text{End}}(m)$, which are all isomorphic independently of the Hilbert space structures we put on the multiplicity spaces.
\end{rem}

We have the following unitary version of Ostrik's theorem:

\begin{thm}\label{unitaryOstrik}(\cite[Theorem 4.6]{doi:10.1063/5.0071215},\ c.f. \cite{MR1976459,MR3933035}). Let $\mathcal{C}$ be a unitary fusion category and $\mathcal{M}$ be an indecomposable unitary module category. Let $m\in \mathcal{M}$ and $Q:=\underline{\text{End}}(m)$.

\begin{enumerate}
\item 
The module category $\mathcal{M}$ is equivalent to the category $\mathcal{C}_{Q}$ of right Q-modules internal to $\mathcal{C}$.
\item 
The dual category $\mathcal{C}^{*}_{\mathcal{M}}$ is equivalent to the unitary fusion category $_{Q} \mathcal{C}_{Q}$ of Q-Q bimodules internal to $\mathcal{C}$.
\end{enumerate}

\end{thm}

\begin{ex}\label{canmodcenter} \textbf{The canonical module and $\mathcal{Z}(\mathcal{C})$}. Let $\mathcal{C}$ be a unitary fusion category, and consider $\mathcal{E}:=\mathcal{C}\boxtimes \mathcal{C}^{mp}$, where $\boxtimes$ denotes the Deligne product of fusion categories \cite{MR3242743}. Then the C*-category $\mathcal{C}$ canonical carries the structure of a left $\mathcal{E}$-module category, with action

$$X\boxtimes Y^{mp}\triangleright Z:=X\otimes Z\otimes Y.$$

It is well known that the dual category $\mathcal{E}^{*}_{\mathcal{C}}\cong \mathcal{Z}(\mathcal{C})$, where $\mathcal{Z}(\mathcal{C})$ denotes the \textit{Drinfeld center} of $\mathcal{C}$ \cite{MR1966525,MR3242743}.

Recall that if $\cC$ is a unitary fusion category, its Drinfeld center $\mathcal{Z}(\cC)$ is a braided unitary fusion category that controls $\cC$'s Morita theory \cite{MR3242743}. We will follow the definition conventions of \cite{MR1966525}, to which we refer the reader for further details on $\mathcal{Z}(\cC)$. Briefly, objects in $\mathcal{Z}(\cC)$ consist of pairs $(Z,\sigma)$, where $Z\in \text{Obj}(\cC)$ and $\sigma=\{\sigma_{Z, X}:Z\otimes X\cong X\otimes Z\ |\ X\in \text{Obj}(\cC)\}$ is a family of unitary isomorphisms, natural in $X$, satisfying the hexagon relation (in $X$). The family $\sigma$ is called a unitary half-braiding. Morphisms $(Z,\sigma)\rightarrow (W,\delta)$ are morphisms $f:Z\rightarrow W$ in $\cD$ that intertwine the half-braidings.

As a consequence of the unitary version of Ostrik's theorem mentioned above, for any object $Y\in \mathcal{C}$, we obtain a Q-system $Q_{Y}:=\underline{\text{End}}(Y)\in \mathcal{E}=\mathcal{C}\boxtimes \mathcal{C}^{mp}$, such that $\mathcal{Z}(\mathcal{C})\cong _{Q_{Y}} \mathcal{E}_{Q_{Y}}$. The object $Q_{\mathbbm{1}}$ is sometimes called the symmetric enveloping algebra object, or the Longo-Rehren algebra.

\end{ex}

\subsubsection{Realization}\label{SectionRealization}

In this section, we tie together actions and Q-systems via the \textit{realization construction}. For the rest of this section, let $\mathcal{C}$ be a unitary fusion category, $A$ a unital simple separable C*-algebra, and suppose we are given a \textit{fully faithful action} $F:\mathcal{C}\rightarrow \textbf{Bim}(A)$. Then for any $Q$-system, we can construct the realization C*-algebra $|Q|$ \cite[Section 4.1]{MR4419534}. This is a unital C*-algebra containing $A$, and comes equipped with a faithful conditional expectation $E_{A}: |Q|\rightarrow A$ with finite Watatani index. This is simply a reflection of the fact that if $A$ is a simple separable C*-algebra, then Q-systems in the C*-tensor category $\textbf{Bim}(A)$ simply \textit{are} finite Watatani index extensions of $A$, thus any $Q$-system in $\mathcal{C}$ can simply be ``pushed forward" to obtain a finite index extension.

In particular, for $Q\cong \bigoplus_{Y\in \text{Irr}(\mathcal{C})} \mathcal{C}(Y, Q)\otimes Y$, then

$$|Q|=F(Q)\cong \bigoplus_{Y\in \text{Irr}(\mathcal{C})} \mathcal{C}(Y, Q)\otimes F(Y)$$

\noindent where now the $\otimes$ is literally the $\otimes$ of vector spaces. The associative product on the C*-algebra $|Q|$ is the pushforward under $F$ of the algebra multiplication morphism for $Q$.

In this situation, we have a linear restriction functor (which is not monoidal in general) $\textbf{Res}: \textbf{Bim}(|Q|)\rightarrow \textbf{Bim}(A)$ defined as follows:

\begin{enumerate}
\item 
For $X\in \textbf{Bim}(|Q|)$, consider $X$ as a vector space, where left and right $A$-actions just the restrictions of $|Q|$ actions. The right $A$-valued inner product is given $\langle \xi\ |\ \nu\rangle_{A}:=E_{A}(\langle \xi\ |\ \nu \rangle_{|Q|})$.
\item 
For a bimodule intertwiner $f$, $\textbf{Res}(f)=f$ is the same linear map thought of as an $A$-bimodule intertwiner.
\end{enumerate}

We note that even though $\langle \xi\ |\ \nu\rangle_{A}\ne \langle \xi\ |\ \nu \rangle_{|Q|}$, since $E_{A}$ has finite Watatani index these inner products induce the same topology on $X$, so $\textbf{Res}(X)$ is indeed a Hilbert module without needing to complete. As a direct corollary of the main result of \cite{MR4419534} (that the 2-category of C*-algebras is Q-system complete), we have the following proposition:

\begin{thm}\label{realizationthm}
Let $\textbf{Bim}(|Q|,\ \mathcal{C})$ be the full tensor subcategory of $\textbf{Bim}(|Q|)$ spanned by objects $X$ such that $\textbf{Res}(X)\cong F(Y)$ for some $Y\in \mathcal{C}$. Then $\textbf{Bim}(|Q|, \mathcal{C})\cong\ _{Q}\mathcal{C}_{Q}$ as C*-tensor categories.
\end{thm}

\subsection{DHR categories for fusion spin chains}

In this section, let $\mathcal{D}$ be a unitary fusion category, and let $X$ be a strongly tensor generating self-dual object. Let $F\mapsto A_{F}$ be the net of algebras on $\mathbbm{Z}$ as in Proposition \ref{fusionspinchain}. Our goal in this section is to analyze $\textbf{DHR}(A)$. We will use the machinery of quantum symmetries described in the previous three subsections to prove that $\textbf{DHR}(A)\cong \mathcal{Z}(\mathcal{D})$ as unitary braided tensor categories.

First, fix any interval $[a,b]\subseteq \mathbbm{Z}$. Notice that the algebra $A_{(\infty,a)}$ by definition is precisely the algebra from the standard action of $\mathcal{D}$ built from $X$ (Example \ref{standardactions}). This gives us a fully faithful unitary tensor functor $L^{a}:\mathcal{D}\rightarrow \textbf{Bim}(A_{(-\infty,a)})$. Similarly, we see that $A_{(b,\infty)}$ is precisely the algebra obtained from the standard action of $\mathcal{D}^{mp}$ with object $X^{mp}$. Thus we have a fully faithful unitary tensor functor $R^{b}: \mathcal{D}^{mp}\rightarrow \textbf{Bim}(A_{(b,\infty)})$. Putting this together, we obtain a fully faithful unitary tensor functor

$$L^{a}\boxtimes R^{b}: \mathcal{D}\boxtimes \mathcal{D}^{mp}\rightarrow \textbf{Bim}(A_{(-\infty,a)}\otimes A_{(b,\infty)}).$$

Here, $L^{a}\boxtimes R^{b}(Y\boxtimes Z^{mp})=L^{a}(Y)\otimes R^{b}(Z^{mp})$, where the tensor product is simply a linear tensor product, and this space is equipped with the obvious structure of an $A_{(-\infty,a)}\otimes A_{(b,\infty)}$ algebraic bimodule. We then complete this with respect to the natural $A_{(-\infty,a)}\otimes A_{(b,\infty)}$-valued inner product. We also remark that the symbol $\otimes$ in $A_{(-\infty,a)}\otimes A_{(b,\infty)}$ is unambiguous, since the two algebras are AF, hence nuclear, as C*-algebras. Since $L^{a}$ and $R^{b}$ are fully-faithful, so is $L^{a}\boxtimes R^{b}$.

Now consider the indecomposable $\mathcal{D}\boxtimes \mathcal{D}^{mp}$-module category $\mathcal{D}$ as described in Example \ref{canmodcenter}. Pick the object $m:=X^{b-a+1}$, and set $Q_{a,b}:=\underline{\text{End}}(m)$ as in Section \ref{modulecatandQsys}. Note that inside $A$, we have a canonical embedding $\Delta_{a,b}:A_{(-\infty,a)}\otimes A_{(b,\infty)}\hookrightarrow A$, given by $$\Delta_{a,b}(f\otimes g):=f\otimes 1_{X^{b-a+1}}\otimes g\in A.$$

One of the primary purposes of Section \ref{QuantumSymms} is to clearly state the following theorem. It is a version of standard results on the symmetric enveloping inclusion/ asmyptotic inclusion/ Longo-Rehren inclusion from subfactor theory (see \cite{MR1302385, MR1642584, MR1332979} respectively). This result is certainly well known to experts, but we could not find it precisely stated in the literature in the form we need. The closest statement to the following that we know of is in \cite[Section 6]{2021arXiv211106378C}.

\begin{thm}\label{TheDrinfeldCenter}
For any interval $[a,b]$, there is an isomoprhism of C*-algebras $|Q_{a,b}|\cong A$ which restricts to $\Delta_{a,b}$ on $A_{(-\infty,a)}\otimes A_{(b,\infty)}$. In particular, we have a fully faithful action $F_{a,b}:\mathcal{Z}(\mathcal{D})\rightarrow \textbf{Bim}(A)$ whose image is characterized as the bimodules of $A$ whose restriction to the subalgebra $A_{(-\infty,a)}\otimes A_{(b,\infty)}$ lie in the image $L^{a}\boxtimes R^{b}(\mathcal{D}\boxtimes \mathcal{D}^{mp})$.
\end{thm}

\begin{proof}
It suffices to construct the isomorphism $|Q_{a,b}|\cong A$. The rest follows immediately from Theorems \ref{unitaryOstrik} and \ref{realizationthm}, where $F_{a,b}$ is the identification of $\mathcal{Z}(\mathcal{D})$ with $\textbf{Bim}(|Q_{a,b}|, \mathcal{D}\boxtimes \mathcal{D}^{mp})$. 

Now, using the description of internal endomorphisms and realizations from Sections \ref{modulecatandQsys} and \ref{SectionRealization}, we see that 

$$|Q_{a,b}|\cong \bigoplus_{Y,Z\in \text{Irr}(\mathcal{D})} \mathcal{D}(Y\otimes X^{b-a+1}\otimes Z, X^{b-a+1})\otimes L^{a}(Y)\otimes R^{b}(Z^{mp}).$$

\noindent But by construction of standard actions, $L^{a}\boxtimes R^{b}$ is an inductive limit of the actions $L^{a}_{n}\boxtimes R^{b}_{n}:\mathcal{D}\boxtimes \mathcal{D}^{mp}\rightarrow \textbf{Bim}(A_{[a-n,a)}\otimes A_{(b,b+n]})$. By \cite[Theorem 4.6]{https://doi.org/10.48550/arxiv.2207.11854}, $|Q_{a,b}|\cong \varinjlim |Q_{a,b}|_{n}$, where the latter denotes the realization with respect to the $L^{a}_{n}\boxtimes R^{b}_{n}$ functors. 

To build the desired isomorphism, we will construct an isomorphism 

$$\pi_{n}:|Q_{a,b}|_{n}\cong A_{[a-n,b+n]}.$$

\noindent Note that $$L^{a}_{n}(Y)=\mathcal{D}(X^{n}, X^{n}\otimes Y)$$ and 

$$R^{b}_{n}(Z^{mp})=\mathcal{D}^{mp}((X^{mp})^{\otimes n}, (X^{mp})^{\otimes n}\otimes Z^{mp})=\mathcal{D}(X^{n}, Z\otimes X^{n}).$$

\bigskip

\noindent For $f\otimes g\otimes h\in \mathcal{D}(Y\otimes X^{b-a+1}\otimes Z, X^{b-a+1})\otimes L^{a}_{n}(Y)\otimes R^{b}_{n}(Z^{mp})$, define $\pi_{n}$ by

$$\pi_{n}(f\otimes g\otimes h):=(1_{X^{n}}\otimes f\otimes 1_{X^{n}})\circ \left(g\otimes 1_{X^{b-a+1}}\otimes h\right).$$
\bigskip

Tracking through the definitions, it is easy to see this is an isomorphism of C*-algebras which is compatible with the local inclusions in the inductive limit. Therefore, it extends to the desired $\pi$. Clearly this restricts to $\Delta_{a,b}$ on $A_{(-\infty, a)}\otimes A_{(b,\infty)}$.

\end{proof}

We note that the above theorem furnishes us with a conditional expectation 

$$E_{A_{(-\infty, a)}\otimes A_{(b,\infty)}}:A\rightarrow A_{(-\infty, a)}\otimes A_{(b,\infty)}$$ 

\noindent by transporting the conditional expectation from the realization $|Q_{a,b}|$. Another immediate consequence of the above theorem is algebraic Haag duality for fusion categorical spin chains.

\begin{prop}\label{Haagduality}
If $X$ strongly tensor generates the fusion category $\mathcal{D}$, the net $A$ constructed above satisfies algebraic Haag duality and uniformly bounded generation.
\end{prop}

\begin{proof}

To see algebraic Haag duality, let $n$ be the smallest positive integer $n$ such that $X^{n}$ contains a copy of every simple. Fix any interval $[a,b]$ with $b-a>n$. 

The relative commutant $Z_{A}(A_{(-\infty,a)}\otimes A_{(b,\infty)})$ corresponds to the central vectors in $A$ as an $A_{(-\infty,a)}\otimes A_{(b,\infty)}$ bimodule. But since $A_{(-\infty,a)}\otimes A_{(b,\infty)}$ is simple and $\mathcal{D}\boxtimes \mathcal{D}^{mp}\rightarrow \textbf{Bim}(A_{(-\infty,a)}\otimes A_{(b,\infty)})$ is fully faithful, the central vectors must lie in the summand isomorphic to copies $A_{(-\infty,a)}\otimes A_{(b,\infty)}$. From the description of Q-system realization from above, this is precisely isomorphic to $\pi(A_{(-\infty,a)}\otimes A_{[a,b]}\otimes A_{(b,\infty)}\subseteq A)$, where $\pi$ is the isomorphism from the previous theorem. But $A_{(-\infty,a)}\otimes A_{(b,\infty)}$ has trivial center and thus the central vectors are of the form $1_{(-\infty,a)}\otimes A_{[a,b]}\otimes 1_{(b,\infty)}$ as desired.

We claim that uniformly bounded generation holds with constant $n+1$, where $n$ is again the smallest positive integer with $X^{n}$ containing copies of all simples.

We will show that if $k\ge n+1$, then the algebra $A_{[a,a+k]}\cong \mathcal{D}(X^{ k+1}, X^{ k+1})$ is generated by the subalgebras $A_{[a,a+k-1]}\cong \mathcal{D}(X^{ k}, X^{ k})\otimes 1_{X} $ and $A_{[a+1,a+k]}\cong 1_{X}\otimes \mathcal{D}(X^{ k}, X^{ k})$. This will imply our desired result inductively.

Since $X^{n}$ contains all simple objects as summands, $X^{l}$ will contain all simple objects as summands for $l\ge n$ by Remark \ref{FrobRem}. By semisimplicity, if we pick, for each triple of isomorphism classes of simple objects $Y,Z,W$, bases $\{e^{Y}_{X^{ k},i}\}$ of $\mathcal{D}(X^{ k}, Y)$,  a basis $\{f^{XZ}_{Y,j}\}$ of $\mathcal{D}(Y, X\otimes Z)$, and a basis $g^{W}_{ZX,l}$ of $\mathcal{D}(Z\otimes X, W)$, then we have the set

$$\{\left(1_{X}\otimes (\ (e^{W}_{X^{k},s})^{*} \circ g^{W}_{ZX,l}\ ) \right)\circ \left((f^{XZ}_{Y,j}\circ e^{Y}_{X^{k},i}\ )\otimes 1_{X}\right) : Y,Z,W\in \text{Irr}(\cC) \}$$

\noindent where the indices $s,l,j,i$ range over all possible values is a basis for $\mathcal{D}(X^{k+1},X^{k+1})$. Therefore it suffices to show any such element is a product $(1_{X}\otimes \alpha)\circ (\beta\otimes 1_{X})$ with $\alpha,\beta\in \mathcal{D}(X^{ k}, X^{ k})$. Since $k\ge n+1$, $X^{ k-1}$ contains all simple objects as summands, there is a nonzero morphism $h\in \mathcal{D}(Z, X^{ k-1})$ with $h^{*}\circ h=1_{Z}$.

Then choosing a specific basis element from above, if we set 

$$\alpha:=((1_{X}\otimes h)\circ f^{XZ}_{Y,j}\circ e^{Y}_{X^{k},i})\otimes 1_{X}\in \mathcal{D}(X^{ k}, X^{ k}), $$ and 

$$\beta:=1_{X}\otimes ( (e^{W}_{X^{k},l})^{*} \circ (g^{W}_{ZX,k}\circ h^{*}\otimes 1_{X}) )\in \mathcal{D}(X^{k}, X^{k}),$$

then

$$(1_{X}\otimes \alpha)\circ (\beta\otimes 1_{X})=\left(1_{X}\otimes ( (e^{W}_{X^{k},l})^{*} \circ g^{W}_{ZX,k} ) \right)\circ \left(( f^{XZ}_{Y,j}\circ e^{Y}_{X^{n},i})\otimes 1_{X}\right).$$

\noindent as desired.

\end{proof}



We return to the actions $F_{a,b}$ of $\mathcal{Z}(\mathcal{D})$ built in Theorem \ref{TheDrinfeldCenter} $F_{a,b}$. Using the AF model for the $Q_{a,b}$ realization, we can explicitly write down an AF model for the functor $F_{a,b}$, by considering the dual actions of $\mathcal{Z}(\mathcal{D})$ to $\mathcal{D}\boxtimes \mathcal{D}^{mp}$ on the finite dimensional algebras $A_{[a-k,b+k]}$. This has essentially been done in \cite[Section 6]{2021arXiv211106378C} with slightly different conventions (and in the $\rm{II}_{1}$ factor framework), but we include details here for the convenience of the reader.

Let $(Z,\sigma)\in \mathcal{Z}(\mathcal{D})$, where $Z\in \mathcal{D}$ and $\sigma=\{\sigma_{Z,Y}\ :\ Y\in \mathcal{D}\}$ is a unitary half-braiding. Then for each interval $I_{k}:=[a-k,b+k]$, we have the $A_{I_k}$ bimodule 

$$F^{k}_{a,b}(Z,\sigma):=\mathcal{D}( X^{ 2k+b-a+1}, X^{ k+b-a+1}\otimes Z\otimes  X^{k})$$

\noindent with right $A_{I_k}$ Hilbert module structure 

$$\langle f | g\rangle_{A_{I_k}}=f^{*}\circ g. $$

The right action is the obvious (pre-composition), while the left action is given by

$$x\triangleright f:= (1_{X^{k+b-a+1}}\otimes \sigma^{*}_{Z,X^{k}})\circ x \circ (1_{X^{k+b-a+1}}\otimes \sigma_{Z,X^{k}})\circ f .$$

If $\xi \in \mathcal{Z}(\cD)((Z,\sigma),(W,\delta))$, then $F^{k}_{a,b}(\xi):F^{k}_{a,b}(Z,\sigma)\rightarrow F^{k}_{a,b}(W,\delta)$ is defined by

$$F^{k}_{a,b}(\xi)(f):= (1_{X^{b-a+k+1}}\otimes \xi\otimes 1_{X^{k}})\circ f. $$

\noindent We have tensorators $(F^{k}_{a,b})^{2}_{(Z,\sigma),(W,\delta)}: F^{k}_{a,b}(Z,\sigma)\boxtimes _{A_{I_k}} F^{k}_{a,b}(W,\delta)\cong F^{k}_{a,b}(Z\otimes W, \sigma \otimes \delta)$ given by

$$(F^{k}_{a,b})^{2}_{(Z,\sigma),(W,\delta)}(f\boxtimes g):= (1_{X^{k+b-a+1}\otimes Z} \otimes \delta^{*}_{W,X^{k}})\circ (f\otimes 1_{W})\circ (1_{X^{k+b-a+1}}\otimes \delta_{W,X^{k}} )\circ g.$$

\noindent These assembles into a unitary tensor functor $F^{k}_{a,b}: \mathcal{Z}(\cD)\rightarrow \textbf{Bim}(A_{I_{k}})$.

We have a natural inclusion $F^{k}_{a,b}(Z,\sigma)\rightarrow F^{k+1}_{a,b}(Z,\sigma)$ given by $f\mapsto 1_{X}\otimes f\otimes 1_{X}$. This is an isometry of Hilbert modules, and is compatible with the $A_{I_{k}}$ and $A_{I_{k+1}}$ actions and bimodule structure in the sense of Definition \ref{cocyclemorphism} (we denote these bimodule inclusions $\nu_{k}$ if the object $(Z,\sigma)\in \mathcal{Z}(\mathcal{D})$ is clear from context). The resulting inductive limit action $\varinjlim_{k} F^{k}_{a,b}$ is an action on $A$, which is canonically monoidally equivalent to the action $F_{a,b}$, so we identify these actions. We will denote the resulting inclusions $j_{a-k,b+k}:F^{k}_{a,b}(Z,\sigma)\hookrightarrow F_{a,b}(Z,\sigma)$.

\begin{lem}\label{GettingPP-bases}
 If $b-a\ge n$, then $F_{a,b}(Z,\sigma)$ has a projective basis localized in $[a,b]$ for any $(Z,\sigma)\in \mathcal{Z}(\mathcal{D})$.
\end{lem}

\begin{proof}

If $b-a\ge n$, then all simple objects occur as a summand of $X^{b-a+1}$. Thus there is a projective basis for $F^{0}_{a,b}(Z,\sigma)$ as a right $A_{[a,b]}$ correspondence. Indeed, pick any finite collection of morphisms $\{b_{i}\}\subset F^{0}_{a,b}(Z,\sigma)=\mathcal{D}(X^{b-a+1},X^{b-a+1}\otimes Z)$ with 

$$\sum_{i} | b_{i}\rangle_{A_{[a,b]}}\  \langle b_{i}| = \sum_{i} b_{i}\circ b^{*}_{i}=1_{X^{b-a+1}\otimes Z}=id_{F^{0}_{[a,b]}}$$

But since the inclusion $F^{0}_{a,b}(Z,\sigma)\hookrightarrow F^{k}_{a,b}(Z,\sigma)$ is a Hilbert module isometry, the image of the $b_{i}$ satisfies 

\begin{align*}
&\sum_{i} | 1_{X^{k}}\otimes b_{i}\otimes 1_{X^{k}}\rangle_{A_{[a-k,b+k]}}\  \langle 1_{X^{k}}\otimes b_{i}\otimes 1_{X^{k}}|\\
&= \sum_{i} (1_{X^{k}}\otimes b_{i}\otimes 1_{X^{k}})\circ (1_{X^{k}}\otimes b^{*}_{i}\otimes 1_{X^{k}})\\
&=1_{X^{2k+b-a+1}\otimes Z}=id_{F^{k}_{[a,b]}(Z,\sigma)}
\end{align*}

Since this is true for all $k$, the image $j_{a,b}(b_{i})$ in the inductive limit $F_{a,b}(Z,\sigma)$ is also a projective basis. Now, to see it satisfies the localization condition, let $x\in A_{[c,d]}\cong \mathcal{D}(X^{d-c+1}, X^{d-c+1})$ with $d<a$. Then to see its action on $j_{a,b}(b_i)$, set $k=a-c$. Then the inclusion of $x$ into $A_{[a-k,b+k]}$ is given by $x\otimes 1_{X^{b+a-c-d}}\in A_{[c,b+a-c]}=A_{[a-k,b+k]}$. We compute

\begin{align*} 
i_{c,d}(x)\triangleright j_{a,b}(b_{i})&=j_{a-k,b+k}(x\otimes 1_{X^{b-d+k}}\triangleright 1_{X^{k}}\otimes b_{i} \otimes 1_{X^{k}})\\
&=j_{a-k,b+k}(x\otimes 1_{X^{a-d}}\otimes 1_{X^{b-c}}\triangleright 1_{X^{k}}\otimes b_{i} \otimes 1_{X^{k}})\\
&=j_{a-k,b+k}((x\otimes 1_{X^{a-d}}\otimes 1_{X^{b-c}})\circ (1_{X^{k}}\otimes b_{i} \otimes 1_{X^{k}}))\\
&=j_{a-k,b+k}( (1_{X^{k}}\otimes b_{i} \otimes 1_{X^{k}})\circ (x\otimes 1_{X^{a-d}}\otimes 1_{X^{b-c}}\circ)\\
&=j_{a,b}(b_{i})\triangleleft i_{c,d}(x)
\end{align*}

Now we check the case for $b<c$, and we set $k=d-b$. Then $[a-k,b+k]$ contains both $[a,b]$ and $[c,d]$. We obtain

\begin{align*} 
i_{c,d}(x)\triangleright j_{a,b}(b_{i})&=j_{a-k,b+k}(1_{X^{c-a+k}}\otimes x\triangleright 1_{X^k}\otimes b_{i}\otimes 1_{X^{k}})\\
&=j_{a-k,b+k}(1_{X^{c-a+k}}\otimes x \triangleright 1_{X^{d-b}}\otimes b_{i} \otimes 1_{X^{d-b-c}})\\
&=(1_{X^{k+b-a}}\otimes \sigma^{*}_{Z,X^{k}})\circ (1_{X^{c-a+k}}\otimes x) \circ (1_{X^{k+b-a}}\otimes \sigma_{Z,X^{k}})\circ(1_{X^{k}}\otimes b_i \otimes 1_{X^{k}})\\
&=j_{a-k,b+k}( (1_{X^{k}}\otimes b_{i} \otimes 1_{X^{k}})\circ (x\otimes 1_{X^{c-a+k}}))\\
&=j_{a,b}(b_{i})\triangleleft i_{c,d}(x)
\end{align*}

In the second to last step we have crucially used naturality of the half-braiding.

\end{proof}

\begin{lem}\label{transportability}
For any two intervals $[a,b]$ and $[c,d]$ of length greater than $n$ and any object $(Z,\sigma)\in \mathcal{Z}(\cD)$, $F_{a,b}(Z,\sigma)\cong F_{c,d}(Z,\sigma)$.
\end{lem}

\begin{proof}
First assume $b\le d$. We recall the building blocks of the inductive limit model $$F^{k}_{a,b}(Z,\sigma):=\mathcal{D}( X^{2k+b-a+1}, X^{ k+b-a+1}\otimes Z\otimes  X^{k}).$$ 

\noindent For a given $k$, choose $m$ such that $[a-k,b+k]\subseteq [c-m,d+m]$. Then we consider the map

$\kappa_{k}: F^{k}_{a,b}(Z,\sigma)\rightarrow F_{c,d}(Z,\sigma)$ by

$$\kappa_{k}(x):=j_{c+m,d-m}((1_{X^{b-c+m}}\otimes \sigma_{Z,X^{d-b}} \otimes 1_{X^{m}})\circ (1_{X^{a-k-c+m}}\otimes x\otimes 1_{X^{d-b-k+m}}))$$

\noindent Note that this does not depend on the choice of $m$. Furthermore by construction, this is $A_{[a-k,b+k]}$ bimodular. In order to show this extends to a well-defined bimodule intertwiner from $F_{a,b}(Z,\sigma)$, we have to show 
 for every $k$, it is compatible with the inclusions $\nu_{k}:F^{k}_{a,b}(Z,\sigma)\rightarrow F^{k}_{a,b}(Z,\sigma)$ in the sense the $\kappa_{k+1}\circ \nu_{k}= \kappa_{k}$. 
Choose $m$ such that $[a-k-1,b+k+1]\subseteq [c-m,d+m]$. Let  for $x\in A_{[a-k,b+k]}$,  recall $\nu_k(x)=1_{X}\otimes x \otimes 1_{X}$. Then we have

\begin{align*}
\kappa_{k+1}\circ \nu_{k}(x)&=j_{c+m,d-m}((1_{X^{b-c+m}}\otimes \sigma_{Z,X^{d-b}} \otimes 1_{X^{m}})\circ (1_{X^{a-k-1-c+m}}\otimes \nu_k(x)\otimes 1_{X^{d-b-k-1+m}}))\\
&=j_{c+m,d-m}((1_{X^{b-c+m}}\otimes \sigma_{Z,X^{d-b}} \otimes 1_{X^{m}})\circ (1_{X^{a-k-c+m}}\otimes x)\otimes 1_{X^{d-b-k+m}}))\\
&=\kappa_{k}(x)
\end{align*}

\noindent Thus by \cite[Proposition 4.4]{https://doi.org/10.48550/arxiv.2207.11854}, the family of $\kappa_{k}$ extend to a bimodule intertwiner $$v:F_{a,b}(Z,\sigma)\cong F_{c,d}(Z,\sigma).$$ Since each $\kappa_{k}$ is an isometry, so is the extension.

We can see that $v$ is a unitary explicitly by exchanging the roles of the intervals $[a,b]$ and $[c,d]$, and using $\sigma^{-1}=\sigma^{*}$ in place of $\sigma$. Incidently, this is also how we build the unitary in the case $d<b$. 

Alternatively, if we first assume $(Z,\sigma)$ is a simple object in $\mathcal{Z}(\mathcal{D})$, then $F_{a,b}(Z,\sigma)$ and $F_{c,d}(Z,\sigma)$ are both simple objects in the C*-category of correspondences, since both $F_{a,b}$ and $F_{c,d}$ are fully faithful. Thus any isometry between them is a unitary, and we obtain the desired result for simple objects. Since $\mathcal{Z}(\mathcal{D})$ is semisimple and $F_{a,b}$ and $F_{c,d}$ respect direct sums, the general result follows.

\end{proof}

\begin{cor}
For any interval $[a,b]$ with $b-a\ge n$, $F_{a,b}(\mathcal{Z}(\mathcal{D}))\subseteq \textbf{DHR}(A)$.
\end{cor}

\begin{proof}
Let $[c,d]$ be any other interval with $b-a\ge n$, and $v:F_{a,b}(Z,\sigma)\cong F_{[c,d]}(Z,\sigma)$ the unitary bimodule isomorphism from the previous lemma. Then there exist a projective basis $\{b_i\}$ localized in $[c,d]$. Then $\{v^{*}(b_{i})\}\subseteq F_{a,b}(Z,\sigma)$. Then 
\begin{align*}\sum v^{*}(b_{i}) \langle v^{*}(b_{i}) | a\rangle_{F_{a,b}(Z,\sigma)} &=\sum   v^{*}(b_{i}) \langle b_{i} | v(a)\rangle_{F_{c,d}(Z,\sigma)}\\
&=\sum   v^{*}(b_{i} \langle b_{i} | v(a)\rangle_{F_{c,d}(Z,\sigma)})\\
&=v^{*}(v(a))=a
\end{align*}

Thus $\{v^{*}(b_{i})\}$ is a projective basis. Now, it is also localized in $[c,d]$ since for any $a\in A_{[c,d]^{c}}$ we have

$$av^{*}(b_i)=v^{*}(ab_i)=v^{*}(b_i a)=v^{*}(b_i)a.$$

\end{proof}

\begin{lem} For any $[a,b]$, the functor $F_{a,b}:\mathcal{Z}(\mathcal{D})\rightarrow \textbf{DHR}(A)$ is braided.
\end{lem}

\begin{proof}
We present an argument which is essentially the same as \cite[Proposition 6.15]{2021arXiv211106378C}. Fix $I:=[a,b]$ with $b-a\ge n$, and let $\{e_{i}\}\subset F^{0}_{a,b}(Z,\sigma)=\mathcal{D}(X^{b-a+1},X^{b-a+1}\otimes Z)$ and $\{f_{j}\}\subset F^{0}_{a,b}(W,d)=\mathcal{D}(X^{b-a+1},X^{b-a+1}\otimes W) $ be projective bases, so that $\{j_{a,b}(e_{i})\}$ and $\{j_{a,b}(f_{j})\}$ are projective bases for $F_{a,b}(Z,\sigma)$ and $F_{a,b}(W,d)$ by Lemma \ref{GettingPP-bases}.

Then it suffices to show 

\begin{equation}\label{braid equation}
\begin{gathered}(F_{a,b})^{2}_{(W,d),(Z,\sigma)}\circ u_{F_{a,b}(Z,\sigma),F_{a,b}(W,d)}(j_{a,b}(e_{i})\boxtimes j_{a,b}(f_{j}))\\
= j_{a,b}\left(1_{X^{b-a+1}}\otimes \sigma_{Z, W})\circ (e_{i}\otimes 1_{W})\circ f_{j}\right)
\end{gathered}
\end{equation}

We compute the left hand side. First pick an interval $[c,d]$ with $b<<c$, and consider a projective basis $\{f^{\prime}_{j}\}$ of the corresponding $A_{[c,d]}$ module $F^{0}_{c,d}(W,d)$, so that $\{v^{*}(j_{c,d}(f^{\prime}_{j}))\}$ is a projective bases of $F_{a,b}(W,d)$ localized in $[c,d]$ as in the proof of Lemma \ref{transportability}.

Then 

\begin{align*}
    &u_{F_{a,b}(Z,\sigma),F_{a,b}(W,d)}(j_{a,b}(e_{i})\boxtimes j_{a,b}(f_{j}))\\
    &=u_{F_{a,b}(Z,\sigma),F_{a,b}(W,d)}\left(\sum_{l} j_{a,b}(e_{i})\boxtimes v^{*}(j_{c,d}(f^{\prime}_{j})) \langle v^{*}(j_{c,d}(f^{\prime}_{j}))\ | j_{a,b}(f_{j}))\rangle\right)\\
    &=\sum_{l} v^{*}(j_{c,d}(f^{\prime}_{j})) \boxtimes j_{a,b}(e_{i})\langle v^{*}(j_{c,d}(f^{\prime}_{j}))\ | j_{a,b}(f_{j}))\rangle\\
    &=\sum_{l,s} j_{a,b}(f_{s})\langle j_{a,b}(f_s)\ |v^{*}(j_{c,d}(f^{\prime}_{j}))\rangle\ \boxtimes\ j_{a,b}(e_{i})\langle v^{*}(j_{c,d}(f^{\prime}_{j}))\ | j_{a,b}(f_{j}))\rangle\\
    &=\sum_{s,l} j_{a,b}(f_{s})\ \boxtimes\ \langle j_{a,b}(f_s)\ |v^{*}(j_{c,d}(f^{\prime}_{j}))\rangle j_{a,b}(e_{i})\langle v^{*}(j_{c,d}(f^{\prime}_{j}))\ | j_{a,b}(f_{j}))\rangle
\end{align*}

We but using the definitions, we see the term 

\begin{align*}
&\sum_{l}\langle j_{a,b}(f_s)\ |v^{*}(j_{c,d}(f^{\prime}_{j}))\rangle j_{a,b}(e_{i})\langle v^{*}(j_{c,d}(f^{\prime}_{j}))\ | j_{a,b}(f_{j})\rangle\\
&=j_{a,b}\left(   (f^{*}_{s}\otimes 1_{Z})  \circ 1_{X^{b-a+1}}\otimes \sigma_{Z,W})\circ (e_{i}\otimes 1_{W})\circ f_{j}\right)
\end{align*}

\noindent Therefore we can evaluate the left hand side of equation \ref{braid equation} to get

\begin{align*}
    &(F_{a,b})^{2}_{(W,d),(Z,\sigma)}\circ u_{F_{a,b}(Z,\sigma),F_{a,b}(W,d)}(j_{a,b}(e_{i})\boxtimes j_{a,b}(f_{j}))\\
    &=j_{a,b}\circ (F^{0}_{[a,b]})^{2}_{(W,d),(Z,\sigma)}\left( \sum_{s} f_{s}\boxtimes (f^{*}_{s}\otimes 1_{Z})  \circ 1_{X^{b-a+1}}\otimes \sigma_{Z,W})\circ (e_{i}\otimes 1_{W})\circ f_{j}\right)\\
    &=j_{a,b}\left((1_{X^{b-a+1}}\otimes \sigma_{Z, W})\circ (e_{i}\otimes 1_{W})\circ f_{j}\right)
    \end{align*}
\end{proof}

\begin{thm}\label{braided equivalence} (c.f. Theorem \ref{ThmC})
For any interval $[a,b]$ with $b-a\ge n$, $F_{a,b}:\mathcal{Z}(\mathcal{D})\rightarrow \textbf{DHR}(A)$ is a braided equivalence.
\end{thm}

\begin{proof}
The only thing left to prove is that $F_{a,b}$ is essential surjectivity onto $\textbf{DHR}(A)$. But since the replete image of $F_{a,b}$ is the same as $F_{c,d}$ in $\textbf{Bim}(A)$, it suffices to show any bimodule $W\in \textbf{DHR}(A)$ is in the image of $F_{c,d}$ for some sufficiently large interval $[c,d]$. Let $W\in \textbf{DHR}(A)$, and choose a basis $\{b_{i}\}$ localized in some in $[c,d]$. By Theorem \ref{TheDrinfeldCenter}, it suffices to show that $W$ lies in the image $L^{c}\boxtimes R^{d}(\mathcal{D}\boxtimes \mathcal{D}^{mp})$ when considered as a $A_{(-\infty,c)}\otimes A_{(d,\infty)}$ bimodule. 

Note that $A$ decomposes as an $A_{(-\infty,c)}\otimes A_{(d,\infty)}$ bimodule via

$$A\cong \bigoplus _{i,j\in \text{Irr}(\cD)} (L^{c}(Y_{i})\boxtimes R^{d}(Y^{mp}_{j}))^{\oplus N_{i,j}}$$

\noindent for some non-negative integers $N_{i,j}$. Now note that since each $[c,d]$-localized basis element $b_{k}\in W$ is $A_{(-\infty,c)}\otimes A_{(d,\infty)}$-central, so the space $b_{k}\triangleleft \left((L^{a}\boxtimes R^{b})(Y_{i}\boxtimes Y^{mp}_{j})\right)$ for each of the $N_{i,j}$ copies of $(L^{a}\boxtimes R^{b})(Y_{i}\boxtimes Y^{mp}_{j})$ in $A$ is a sub $A_{(-\infty,c)}\otimes A_{(d,\infty)}$ bimodule of $W$, and the span as these range over all localized basis elements $b_{k}$ and all $i,j$ is all of $W$. 

But the map $$(L^{a}\boxtimes R^{b})(Y_{i}\boxtimes Y^{mp}_{j})\mapsto b_{k}\triangleleft \left((L^{a}\boxtimes R^{b})(Y_{i}\boxtimes Y^{mp}_{j})\right)$$ is a bounded algebraic bimodule intertwiner, hence is an intertwiner of correspondences $(L^{a}\boxtimes R^{b})(Y_{i}\boxtimes Y^{mp}_{j})\rightarrow W$. But $(L^{a}\boxtimes R^{b})(Y_{i}\boxtimes Y^{mp}_{j})$ is irreducible by fully faithfulness of $L^{a}\boxtimes R^{b}$, so the above map is either a scalar multiple of an isometry (in which case $(L^{a}\boxtimes R^{b})(Y_{i}\boxtimes Y^{mp}_{j})$ is isomorphic to its image) or $0$. But the images of these maps span $X$, and since $X$ itself is semisimple, the images of $(L^{a}\boxtimes R^{b})(Y_{i}\boxtimes Y^{mp}_{j})$ exhaust possible simple summands of $W$.

Thus when we restrict $W$ to an $A_{(-\infty, c)}\otimes A_{(d,\infty)}$ bimodule, then $W$ is a direct sum of the $(L^{a}\boxtimes R^{b})(Y_{i}\boxtimes Y^{mp}_{j})$, hence in the image of $F_{[c,d]}$ as claimed.

\end{proof}

We can immediately use this to approach the problem described in the introduction of distinguishing  quasi-local algebras up to bounded spread isomorphism. The following example is the standard example of global symmetry: spin flips.

\begin{ex}{\textbf{Ordinary spin system}}. Let $d\in \mathbbm{N}$ and consider the onsite Hilbert space $\mathbbm{C}^{d}$, which we view as having a trivial onsite categorical symmetry. The fusion category is $\textbf{Hilb}_{f.d.}$, and the object $X=\mathbbm{C}^{d}$ is clearly strongly tensor generating.

The resulting net $A_{d}$ over $\mathbbm{Z}$ is then the usual net of all local operators, and the quasi-local algebra is the UHF C*-algebra $M_{d^{\infty}}$. By Theorem \ref{braided equivalence}, we have $\textbf{DHR}(A_{d})\cong \mathcal{Z}(\textbf{Hilb}_{f.d.})\cong \textbf{Hilb}_{f.d.}$ as braided tensor categories. 
\end{ex}

\begin{ex}\label{genspinflip}{\textbf{Generalized spin flip}}. Let $G$ be an abelian finite group. Consider the onsite Hilbert space $K:=\mathbbm{C}^{|G|}$, and the action of $G$ on $K$ which permutes the standard basis vectors, i.e. the left regular representation. $K$, as an object in $\textbf{Rep}(G)$, contains all isomorphism classes of simples, hence is a strongly tensor generating object in $\textbf{Rep}(G)$ (with $n=1$). 

Then we consider the net of symmetric observables constructed as above, which we denote $A^{G}$. It is easy to see that the resulting UHF algebra is $M_{|G^{\infty}|}$. In particular as C*-algebras, we have an isomorphism of the quasi-local algebras $A^{G}\cong A_{|G|}$. In particular, for any groups $G$ and $H$ of the same order $A^{G}\cong A^{H}$.

However, by Theorem \ref{braided equivalence}, $\textbf{DHR}(A^{G})\cong \mathcal{Z}(\textbf{Rep}(G))$. This implies that even though $A^{\mathbbm{Z}/2\mathbbm{Z}}\cong A_{2}$ as C*-algebras, there is no isomorphism with bounded spread between these. Similarly, at the level of algebras $A^{\mathbbm{Z}/4\mathbbm{Z}}\cong A^{\mathbbm{Z}/2\mathbbm{Z}\times \mathbbm{Z}/2\mathbbm{Z}}$, but there is no bounded spread isomorphism between these because the underlying fusion categories $\textbf{Hilb}_{f.d.}(\mathbbm{Z}/4\mathbbm{Z})$ and $\textbf{Hilb}_{f.d}(\mathbbm{Z}/2\mathbbm{Z}\times \mathbbm{Z}/2\mathbbm{Z})$ are not Morita equivalent.

\end{ex}

\bigskip

\subsection{Examples of QCA}

Recall that if $\mathcal{C}$ is a braided C*-tensor category, $\textbf{Aut}_{br}(\mathcal{C})$ is the group of  unitary isomorphism monoidal natural isomorphism classes of braided (unitary) monoidal equivalences of $\mathcal{C}$. We have the following corollary to the previous section.

\begin{cor} (c.f. Theorem \ref{ThmC})
Let $A$ denote the fusion spin chain constructed from a fusion category $\mathcal{D}$ and strongly tensor generating object $X$. Then there is a homomorphism $\textbf{DHR}:\textbf{QCA}(A)/\textbf{FDQC}(A)\rightarrow \textbf{Aut}_{br}(\mathcal{Z}(\mathcal{D}))$.
\end{cor}

The goal in this section is to find examples of QCA that map onto specific braided autoequivalences of the center. Let $\cD$ be a unitary fusion category and $X$ a strongly tensor generating object, and let $A$ denote the net over $\mathbbm{Z}$.  Note that any unitary autoequivalence of $\mathcal{C}$ induces a braided unitary autoequivalence of the center $\widetilde{\alpha}\in \text{Aut}_{br}(\mathcal{Z}(\cD))$ \cite{MR3242743}. More specifically, if $(Z, \sigma)\in \mathcal{Z}(\cD)$ and $\alpha\in \text{Aut}_{\otimes}(\cC)$, then define

$$(\alpha(Z), \sigma^{\alpha}),$$

\noindent where $\sigma^{\alpha}_{\alpha(Z), X}: \alpha(Z)\otimes X\cong X\otimes \alpha(Z)$ is defined as the composition

\[
\begin{tikzcd}
\alpha(Z)\otimes X \arrow{r}{can}  & \alpha(Z\otimes \alpha^{-1}(X))\arrow{r}{\alpha(\sigma_{Z,\alpha^{-1}(X)})}& \alpha(\alpha^{-1}(X)\otimes Z)\arrow{r}{can} & X\otimes \alpha(Z).
\end{tikzcd}
\]

\noindent Here, $\text{can}$ denotes the canonical isomorphisms built from the monoidal structure on the functor $\alpha$. It is easy to check that the assignment $\widetilde{\alpha}(Z,\sigma):=(\alpha(Z), \sigma^{\alpha})$ extends naturally to a braided monoidal equivalence of $\mathcal{Z}(\cD)$. Then $\alpha\mapsto \widetilde{\alpha}$ gives a homomorphism from $\textbf{Aut}_{\otimes}(\cD)\rightarrow \textbf{Aut}_{br}(\mathcal{Z}(\cD))$, whose image is denoted $\textbf{Out}(\cD)$. 

Let $\text{Stab}(X)$ be the group whose objects are monoidal equivalence classes of unitary monoidal autoequivalences of $\cD$ such that $\alpha(X)\cong X$. For any $\alpha\in \text{Stab}(X)$, we will build a QCA on $A$ with spread $0$, whose induced action on $\mathcal{Z}(\cD)$ is given by $\widetilde{\alpha}$. Recall $A_{[a,b]}:=\cD(X^{b-a+1}, X^{b-a+1})$. Then applying $\alpha$ to the morphisms in $\cD$ and conjugating by the tensorator of $\alpha$, we get the map

$$\widehat{\alpha}: \cD(X^{b-a+1}, X^{b-a+1})\mapsto \cD(\alpha(X)^{b-a+1}, \alpha(X)^{b-a+1}) $$

Choosing an isomorphism $\eta:\alpha(X)\cong X$, and define the homomorphism

$$Q^{[a,b]}_{\alpha}: A_{[a,b]}\cong \cD(X^{b-a+1}, X^{b-a+1}) \rightarrow \cD(X^{b-a+1}, X^{b-a+1})\cong A_{[a,b]}  $$

by 

$$Q^{[a,b]}_{\alpha}(f):= (\eta^{\otimes b-a})\circ \widehat{\alpha}(f)\circ ((\eta^{*})^{\otimes b-a+1})$$

These isomorphisms are clearly compatible with inclusions, and thus assemble into a QCA with $0$-spread

$$Q_{\alpha}: A\rightarrow A.$$

\noindent By construction, the assignment only depends on the choice of $\eta$ up to a depth one circuit. Furthermore, from our analysis in the previous section, it is clear that $\textbf{DHR}(Q_{a})\cong \widetilde{\alpha}\in \textbf{Out}(\cD)\le \mathcal{Z}(\cD)$.

\begin{cor} (c.f. Theorem \ref{ThmC})
Suppose that the tensor generating object $X$ is stable under any monoidal autoequivalence of $\cC$. Then the image of the $\textbf{DHR}$ homomorphism $\textbf{QCA}(A)/\textbf{FDQC}(A)\rightarrow \textbf{Aut}_{br}(\mathcal{Z}(\cD))$ contains the subgroup $\textbf{Out}(\cD)$. In particular, if $\textbf{Out}(\cD)$ is non-abelian, then so is $\textbf{QCA}(A)/\textbf{FDQC}(A)$.
\end{cor} 

\begin{ex}\label{non-abelian}{\textbf{Non-abelian $\mathbbm{Z}/2\mathbbm{Z}\times \mathbbm{Z}/2\mathbbm{Z} $-symmetric QCA}} (c.f. Corollary \ref{CorD}). We now give a concrete example. We consider on ordinary spin system, coarse-grained so that the on-site Hilbert space consists of two qubits $$K:=\mathbbm{C}^{2}\otimes \mathbbm{C}^{2}$$

Let $G:=\mathbbm{Z}/2\mathbbm{Z}\times \mathbbm{Z}/2\mathbbm{Z} $ act on $K$ where each copy of $\mathbbm{Z}/2\mathbbm{Z}$ acts by a spin flip on the corresponding tensor factor. This defines a global, on-site symmetry. Viewing $K\in \textbf{Rep}(\mathbbm{Z}/2\mathbbm{Z}\times \mathbbm{Z}/2\mathbbm{Z})$, we see that $K$ is in fact the regular representation, and thus is characteristic (since it decomposes as a direct sum of all simple objects with multiplicity 1). 

Thus, the image of $\textbf{DHR}$ for the resulting net contains $$\textbf{Out}(\textbf{Rep}(\mathbbm{Z}/2\mathbbm{Z}\times \mathbbm{Z}/2\mathbbm{Z}))\supseteq \textbf{Out}(\mathbbm{Z}/2\mathbbm{Z}\times \mathbbm{Z}/2\mathbbm{Z})\cong S_{3}. $$

In this case, we can implement this $S_{3}$ action explicitly on the original Hilbert space. Using the standard qubit basis, we consider the basis for $\mathbbm{C}^{2}\otimes \mathbbm{C}^{2}$, we consider the vectors in $\mathbbm{C}^{2}$

$$|+\rangle:= \frac{1}{\sqrt{2}}(|0\rangle+|1\rangle)$$

$$|-\rangle:= \frac{1}{\sqrt{2}}(|0\rangle+|1\rangle)$$

and define the orthonomal basis of $\mathbbm{C}^{2}\otimes \mathbbm{C}^{2}$

$$|a\rangle:= |+\rangle \otimes |+\rangle  $$
$$|1\rangle:= |+\rangle \otimes |-\rangle  $$
$$|2\rangle:= |-\rangle \otimes |+\rangle  $$
$$|3\rangle:= |-\rangle \otimes |-\rangle  $$

Then for any $g\in S_{3}$, consider the unitary  $U_{g}$ on $\mathbbm{C}^{2}\otimes \mathbbm{C}^{2}$,  which fixes $|a\rangle$ and permutes $\{|1\rangle, |2\rangle, |3\rangle\}$ accordingly.

Then conjugation by the product of $\text{Ad}(U_g)$ over all sites gives a spread $0$ QCA on the algebra of $\mathbbm{Z}/2\mathbbm{Z}\times \mathbbm{Z}/2\mathbbm{Z}$ symmetric operators, which cannot be disentangled by a symmetric finite circuit. Note that even though this QCA is defined on the full spin system and preserves the symmetric subalgebra, it does not commute with the group action.

\end{ex}

\subsection{2+1D topological boundaries theories}

Let $\cC$ be a unitary modular tensor category. Recall a \textit{Lagrangian} algebra is a commutative, connected separable algebra object (or Q-system) $A\in \cC$ such that $\text{dim}(A)^{2}=\text{dim}(\mathcal{C})$. Equivalently, the category of local modules $\cC^{loc}_{A}\cong \textbf{Hilb}_{f.d}$. In this case, the category $\cC_{A}$ of right $A$-modules is a fusion category, and the central functor $\cC\rightarrow \mathcal{Z}(\cC_{A})$ is a braided equivalence. We can view Lagrangian algebras as parameterizing ``ways $\cC$ can be realized as a Drinfeld center of a fusion category", and the fusion category in question is $\cC_{A}$.

From a physical perspective, if we view $\cC$ as the topological order of a 2+1D theory, then topological (gapped) boundaries are characterized by Lagrangian algebras $A\in \mathcal{C}$ \cite{MR3246855, MR3063919}. The fusion category $\cC_{A}$ is the fusion category of topological boundary excitations.

 We define the groupoid $\textbf{TopBound}_{2+1}$ as follows:

    \begin{itemize}
        \item 
        Objects are pairs $(\cC,A)$ where $\cC$ is a unitary modular tensor category and $A$ is a Lagrangian algebra.
        \item
        Morphisms $(\cC,A)\rightarrow (\cD,B)$ are pairs $(\alpha, \eta)$, where $\alpha:\cC\cong \cD$ is a unitary braided equivalence and $\eta:\alpha(A)\cong B$ is a unitary isomorphism of algebra objects. These morphisms are taken up to the equivalence relation $(\alpha, \eta)\sim (\beta, \lambda)$ if there is a monoidal natural isomorphism $\delta: \alpha\cong \beta$ such that $\lambda \circ \delta_{A}=\eta$ (see  \cite[Definition 4.1]{MR3933137}.)
        Composition is induced from the natural composition of autoequivalences. 
    \end{itemize}

In this section, we will give a construction of a 1D net of algebras from the data of the pair $(\cC, A)$ which is functorial from $\textbf{TopBound}_{2+1}\rightarrow \textbf{Net}_{\mathbbm{Z}}/\sim_{\textbf{FDQC}}$. Recall that there is a forgetful functor $\textbf{Forget}: \textbf{TopBound}_{2+1}\rightarrow \textbf{BrTens}$ that simply forgets the choice of Lagrangian algebra. We have the following theorem, which allows us to realize many examples of braided equivalences between concrete quasi-local algebras.

\begin{thm}
    There exists a functor $\textbf{B}:\textbf{TopBound}_{2+1}\rightarrow \textbf{Net}_{\mathbbm{Z}}$ such that $\textbf{DHR}\circ \textbf{B}\cong \textbf{Forget}$ as functors $\textbf{TopBound}_{2+1}\rightarrow \textbf{BrTens}$.
\end{thm}

\begin{proof}

    To build $G$, let $(\cC,A)\in \textbf{TopBound}_{2+1} $. Choose the object 
    $$X_{\cC}=\bigoplus_{Z\in \text{Irr}(\cC)} Z\in \cC $$. Then as described above $\cC\cong \mathcal{Z}(\cC_{A})$, and we have a forgetful functor $F_{A}:\cC\rightarrow \cC_{A}$, which is equivalent to the free module functor $Z\mapsto Z\otimes A$.
    We consider fusion spin chain net as in Proposition \ref{fusionspinchain} with fusion category $\cC_{A}$ and generator $F_{A}(X_{\cC})$. This is a strong tensor generator for $\cC_{A}$ since the forgetful functor $F_{A}$ from the center is always dominant, so all simple objects in $\mathcal{C}_{A}$ are already summands of $F_{A}(X_{\cC})$..

    We denote this net over $\mathbbm{Z}$ by $\textbf{B}(\cC,A)$. For an interval with $n$ points, the local algebra is 
    
    $$\cC_{A}(F_{A}(X_{\cC})^{ n},F_{A}(X_{\cC})^{ n})\cong \cC(X^{ n}_{\cC}, X^{ n}_{\cC}\otimes A)$$

   \noindent (see, e.g. \cite{MR3687214}). In the latter model, composition is given by 
    
$$f\cdot g= (1_{X^{n}_{\cC}}\otimes m)\circ (f\otimes 1_{A})\circ g,$$ where $m:A\otimes A\rightarrow A$ is the multiplication. The inclusions $$\cC_{A}(F_{A}(X_{\cC})^{n},F_{A}(X_{\cC})^{n})\hookrightarrow \cC_{A}(F_{A}(X_{\cC})^{ n+1},F_{A}(X_{\cC})^{ n+1})$$ \noindent given by tensoring $1_{F_{A}(X_{\cC})}$ on the left and right are given in our alternate model by sending, for $f\in \cC(X^{\ n}_{\cC}, X^{n}_{\cC}\otimes A) $,
    
$$f\mapsto 1_{X_{\cC}}\otimes f$$
and
$$f\mapsto (1_{X^{n}_{\cC}}\otimes \sigma_{A,X_{\cC}})\circ (f\otimes 1_{X_{\cC}})$$

\noindent respectively, where $\sigma_{A,X_{\cC}}:A\otimes X_{\cC}\cong X_{\cC}\otimes A$ is the braiding in $\cC$.

We will now extend the assignment $(\cC,A)\mapsto \textbf{B}(\cC,A)$ to a functor. Suppose 
$$(\cC,A),(\cD,B)\in \textbf{TopBound}_{2+1}  $$ \noindent and let $\alpha: \cC\cong \cD$ be a braided equivalence with $\alpha(A)\cong B$ as algebra objects. Choose a specific (unitary) algebra isomorphism $\eta: \alpha(A)\cong B$. Then $\alpha(X_{\cC})\cong X_{\cD}$ since both are simply a direct sum of all the simple objects. Pick a such a unitary and call it $\nu$.

Then, using the monoidal structure on $\alpha$, we get an algebra homomorphism

$$\widehat{\alpha}:\cC(X^{n}_{\cC}, X^{n}_{\cC}\otimes A)\mapsto \cD(\alpha(X_{\cC})^{n},\alpha(X_{\cC})^{n}\otimes \alpha(A)).$$

Then for $f\in \cC(X^{n}_{\cC}, X^{n}_{\cC}\otimes A)  \textbf{B}(\cC,A)_{[a,a+n]}$, we define

$$\textbf{B}(\alpha)(f):= ((\nu^{*})^{\otimes n}\otimes \eta)\circ \widehat{\alpha}(f)\circ ((\nu^{*})^{\otimes n})\in \cD(X^{n}_{\cD}, X^{n}_{\cD}\otimes B).$$

Since $\eta$ is an algebra isomorphism, it is easy to see that this is an isomorphism of local algebras. Since $\alpha$ is a braided monoidal equivalence, this is compatible with the left and right inclusions and thus extends to an isomorphism of quasi-local algebras with spread $0$. This gives us a morphism  $\textbf{B}(\alpha)\in \textbf{Net}_{\mathbbm{Z}}(\textbf{B}(\cC,A),\textbf{B}(\cD,B))$. By construction this only depends on the choice of $\eta$ up to a finite depth (in fact, depth one) circuit.

Now, consider $\textbf{DHR}\circ \textbf{B}: \textbf{TopBound}_{2+1}\rightarrow \textbf{BrTens}$. But since $F_{A}:\cC\rightarrow \cC_{A}$ factors through an equivalence with the center $\widetilde{F}_{A}:\cC\rightarrow \mathcal{Z}(\cC_{A})$ \cite{MR3039775}. Thus from the previous section, we have an equivalence $\textbf{DHR}(\textbf{B}(\cC,A))\cong \cC$, and under this identification, $\textbf{DHR}(\textbf{B}(\alpha))=[\alpha]$. 

\end{proof}

\bibliographystyle{alpha}
{\footnotesize{
\bibliography{bibliography}}}
\end{document}